\documentclass[11pt]{article}
\usepackage[letterpaper,margin=1in]{geometry}
\usepackage{xcolor}
\usepackage[colorlinks]{hyperref}
\definecolor{lightblue}{rgb}{0.5,0.5,1.0}
\definecolor{darkred}{rgb}{0.5,0,0}
\definecolor{darkgreen}{rgb}{0,0.5,0}
\definecolor{darkblue}{rgb}{0,0,0.5}
\hypersetup{colorlinks,linkcolor=darkred,filecolor=darkgreen,urlcolor=darkred,citecolor=darkblue}
\usepackage{subcaption}
\definecolor{lightblue}{rgb}{0.5,0.5,1.0}
\definecolor{darkred}{rgb}{0.5,0,0}
\definecolor{darkgreen}{rgb}{0,0.5,0}
\definecolor{darkblue}{rgb}{0,0,0.5}

\usepackage{tabularx}
\usepackage{mathtools}
\usepackage{verbatim}
\usepackage{enumitem}
\usepackage{amsthm}
\usepackage{amssymb}
\usepackage{amsmath}
\usepackage{complexity}
\usepackage{nicefrac}
\usepackage{thmtools}
\usepackage{thm-restate}
\usepackage{mathtools}
\usepackage[linesnumbered,ruled,commentsnumbered,longend]{algorithm2e}
\usepackage{tikz}
\usetikzlibrary{graphs}
\usetikzlibrary{shapes.misc, positioning}
\usetikzlibrary{shapes,fit}
\usepackage[edges]{forest}
\usetikzlibrary{shapes,arrows,calc,decorations.pathmorphing}
\usepackage{amsmath}
\usepackage{cite}

\usepackage[disable]{todonotes}
\newcommand{\pascal}[1]{\todo[color=red!40,size=\footnotesize]{\textbf{P:} #1}}

\newcommand{\markus}[1]{\todo[color=orange!40,size=\footnotesize]{\textbf{M:} #1}}

\newcommand{\Traces}{\textsc{Traces}}
\newcommand{\nauty}{\textsc{nauty}}
\newcommand{\bliss}{\textsc{bliss}}
\newcommand{\saucy}{\textsc{saucy}}
\newcommand{\dejavu}{\textsc{dejavu}}
\newcommand{\conauto}{\textsc{conauto}}

\newcommand{\xparagraph}[1]{\textbf{#1}}

\newtheorem{lemma}{Lemma}

\newtheorem{theorem}[lemma]{Theorem}

\SetKwProg{Fn}{function}{}{end}\SetKwFunction{RandomAut}{Automorphisms}
\SetKwProg{Fn}{function}{}{end}\SetKwFunction{RandomIso}{Isomorphism}
\SetKwFunction{FirstLeaf}{FirstLeaf}
\SetKwFunction{RandomWalk}{RandomWalk}
\SetKwFunction{NextNodeDFS}{OneNodeDFS}
\SetKwFunction{NextNodeBFS}{OneNodeBFS}
\SetKwFunction{Some}{Some}
\SetKwFunction{Sift}{Sift}
\SetKwFunction{RandomElement}{RandomElement}
\SetKwFunction{RandomChild}{RandomChild}
\SetKwFunction{NextChild}{NewChild}
\SetKwFunction{NextRandomChild}{NewRandomChild}
\SetKwFunction{Refine}{Refine}
\SetKwFunction{Subtree}{Subtree}
\SetKwFunction{RRef}{Ref}
\SetKwFunction{SSel}{Sel}
\SetKwFunction{Sift}{Sift}
\SetKwFunction{CertifyAutomorphism}{CertifyAutomorphism}
\SetKwFunction{CertifyIsomorphism}{CertifyIsomorphism}



\DeclareMathOperator{\Sym}{Sym}
\DeclareMathOperator{\Aut}{Aut}
\DeclareMathOperator{\Refx}{Ref}
\let\Sel\undefined
\DeclareMathOperator{\Sel}{Sel}

\DeclareMathOperator{\Inv}{Inv}

\DeclareMathOperator{\AND}{AND}

\newtheorem{fact}{Fact}
{\bfseries}{\itshape}

\title{A Characterization of Individualization-Refinement Trees}
\author{Markus Anders \and Jendrik Brachter \and Pascal Schweitzer}

\newcommand\blfootnote[1]{%
  \begingroup
  \renewcommand\thefootnote{}\footnote{#1}%
  \addtocounter{footnote}{-1}%
  \endgroup
}
\begin{document}

\maketitle

\begin{abstract}
Individualization-Refinement (IR) algorithms form the standard method and currently the only practical method for symmetry computations of graphs and combinatorial objects in general. Through backtracking, on each graph an IR-algorithm implicitly creates an IR-tree whose order is the determining factor of the running time of the algorithm.

We give a precise and constructive characterization which trees are IR-trees. This characterization is applicable both when the tree is regarded as an uncolored object but also when regarded as a colored object where vertex colors stem from a node invariant.
We also provide a construction that given a tree produces a corresponding graph whenever possible. This provides a constructive proof that our necessary conditions are also sufficient for the characterization.
\end{abstract}\blfootnote{Supported by 
the European Research Council (ERC) under the European Union's Horizon 2020 research and innovation programme (EngageS: grant No.~{820148}).}
\markus{walk vs. path}
\markus{node $n$ in a cell vs. vertex $v$ in a cell}
\pascal{class vs cell}
\section{Introduction}
The individualization-refinement (IR) framework is a general backtracking technique employed by algorithms solving tasks related to the computation of symmetries of combinatorial objects~\cite{McKay201494}. These include algorithms computing automorphism groups, isomorphism solvers, canonical labeling tools used for computing normal forms, and to some extent recently also machine learning computations in convolutional neural networks~\cite{DBLP:conf/ijcai/AbboudCGL21,DBLP:conf/aaai/0001RFHLRG19}. In fact all competitive graph isomorphism/automorphism solvers, specifically \nauty/\Traces~\cite{nautyTracesweb, McKay201494},  \bliss~\cite{bliss:webpage, DBLP:conf/alenex/JunttilaK07}, \saucy~\cite{saucy:webpage, Darga:2004:ESS:996566.996712}, \conauto~\cite{conauto:webpage, DBLP:conf/wea/Lopez-PresaCA13}, and \dejavu~\cite{dejavu:webpage, DBLP:conf/alenex/AndersS21} fall within the framework. 
These tools alternate color-refinement techniques (such as the 1-dimensional Weisfeiler-Leman algorithm) with backtracking steps. The latter perform artificial individualization of indistinguishable vertices. This leads to recursive branching and overall to a tree of recursive function calls, the so called IR-tree.

Using clever invariants and heuristics, the tools manage to prune large parts of the IR-tree. Since the non-recursive work is quasi-linear, it has long been known that the number of traversed nodes of the IR-tree is the determining factor in the running time for all the tools (see for example \cite[Theorem 9]{thesis} and \cite{DBLP:journals/corr/abs-0804-4881}). 
And in fact, the running times of the various tools closely reflect this~\cite{McKay201494,DBLP:conf/alenex/AndersS21}. Indeed, variation in the traversal strategies among the tools leads to a different number of traversed nodes which in turn leads to different running times. However, explicit bounds that rigorously show asymptotic advantages of randomized traversals over deterministic ones have only recently been obtained~\cite{anders_et_al:LIPIcs.ICALP.2021.16}. For this, a specific problem --- a search problem in trees with symmetries -- is defined. It captures precisely the parameters within which IR-algorithms operate. 

While these results are quite general within an abstract model, the bounds proven in~\cite{anders_et_al:LIPIcs.ICALP.2021.16} apply to the search problem in arbitrary trees with symmetries, independent of whether they originate from actual IR-computations or not. Granted, the vast benchmark library of \Traces{}~\cite{nautyTracesweb,McKay201494} shows that IR-trees come in an abundance of forms and shapes. However, to date there have been no comprehensive results actually 
analyzing which trees can arise as an IR-tree.

\xparagraph{Contribution.} 
In this paper we study which trees are IR-trees. Arising from a branching process, all IR-trees are rooted and all inner vertices have at least 2 children. Such trees are called irreducible (or series reduced). Despite
a vast variety of IR-trees arising from benchmark libraries,
it turns out that not all irreducible trees are IR-trees. However, we can 
give a full, constructive characterization of IR-trees. 

\begin{theorem} \label{thm:sufficient}
An irreducible tree is an IR-tree if and only if there is no node that has exactly two children of which exactly one is a leaf.
\end{theorem}
To prove the theorem we first provide and justify necessary conditions for a tree to be an IR-tree. We then prove that, indeed, these conditions are sufficient by providing graphs on which the execution of an IR-algorithm yields the desired tree. In fact, our proof is constructive, meaning that we obtain an algorithm with the following property. Given a tree~$T$ satisfying the necessary conditions, the algorithm produces a graph whose IR-tree is~$T$.

As we describe in our definition of IR-trees in Section~\ref{sec:ir}, the trees are naturally associated with a coloring of the vertices. This coloring is a crucial component that is related to the automorphism group structure of the graph. Our characterization also fully describes how color classes may be distributed in a given tree. It turns out that there are several simple restrictions, in particular for vertices that have precisely two children, but apart from that all colorings can be realized and in particular any number of symmetries can be ensured (see Section~\ref{sec:graph:constr}).

Our characterization provides a fundamental argument transferring the analysis of abstract tree traversal strategies performed in \cite{anders_et_al:LIPIcs.ICALP.2021.16} to backtracking trees of IR-algorithms on actual instances.
Specifically, we may conclude that the abstract trees used for the lower bounds of probabilistic algorithms in \cite{anders_et_al:LIPIcs.ICALP.2021.16} indeed appear as IR-tees.
However, interestingly, the abstract trees used for the lower bounds of deterministic algorithms
(Theorem~13, \cite{anders_et_al:LIPIcs.ICALP.2021.16}) are not IR-trees. In fact these trees have nodes with two children, one child that is a leaf and another that is not. This breaks the necessary conditions as laid out by Theorem~\ref{thm:sufficient}.  
Fortunately, it also immediately follows from our results that a slight modification can rectify this: by simply replacing the respective leaves with inner nodes that have two attached leaves, the trees become actual IR-trees, due to our characterization. 
Overall, we therefore prove that the lower bounds of \cite{anders_et_al:LIPIcs.ICALP.2021.16} hold true in the IR-paradigm.

\xparagraph{Cell Selectors and Invariants.} Formally, the IR-paradigm allows for different design choices in some of its components.
For most of these, competitive practical solvers actually make very similar choices:
the refinement is always \emph{color refinement} and solvers commonly choose as their pruning invariant (essentially) the so-called \emph{quotient graph}.    
The way in which the actual implementations differ from color refinement and quotient graphs is usually only in minor details and done to achieve practical speed-ups. This only leads to a slightly weaker refinement and invariants in some specific cases.
In this paper, we therefore comply with these common design choices.

Many other design choices, such as \emph{how} IR-trees are traversed, have no effect on the characterization of the IR-trees themselves.

There is however one integral design choice where competitive IR-solvers do indeed vary in a way that affects which trees are IR-trees, namely 
the so-called \emph{cell selectors}.
We should emphasize that Theorem~\ref{thm:sufficient} 
only says that for the trees satisfying the necessary
conditions there is some cell selector for which the graph is an IR-tree.

However, we can also say something about specific cell selectors. Considering the characterization for a \emph{given} cell selector, there are two possibilities: either, fewer trees turn out to be IR-trees or the same characterization applies.
We can use our results to argue that for some cell selectors that are used in practice our necessary conditions are sufficient, while for others they are not (see Section~\ref{sec:conclusion} for a discussion).

\xparagraph{Techniques.} 
Many properties of a graph, e.g.~symmetries, are directly tied to properties of its IR-tree.
When modeling a graph that is supposed to produce a particular IR-tree, two major difficulties arise, roughly summarized as follows:
\begin{enumerate}
\item The effect of color refinement on the graph needs to be kept under control.
\item The shape of the IR-tree may dictate that symmetries must be simultaneously represented in distinct parts of the graph.
\end{enumerate}
We resolve these issues using various gadget constructions specifically crafted for this purpose.
We introduce \emph{concealed edges}, which allow us to precisely control the point in time at which the IR-process is able to see a certain set of edges and thus color refinement to take effect (resolving issue (1)). 
By combining concealed edges with gadgets enforcing particular regular abelian automorphism groups we can synchronize symmetries across multiple branches of the tree (resolving issue (2)). 

Here, as the main tool we show the following. As an additional restriction, which stems from the structure of IR-trees, we consider only trees where all leaves can be mapped to the same number of other leaves via symmetries (i.e., under automorphisms all \emph{leaf orbits} have the same size). We show that each such tree $T$ can be embedded into a graph $H_T$, such that~$H_T$ restricts the symmetries of~$T$ in a particular way. Intuitively, we keep just enough symmetries to allow leaves to be mapped to each other whenever this is possible in~$T$. We thereby effectively couple leaf orbits 
so that when fixing one leaf, all other leaves are fixed as well. More formally we prove the following theorem.

\begin{restatable}{theorem}{thmtwo}\label{thm:restrict_symmetries_of_tree}
Let~$T$~be a colored tree in which all leaf orbits have the same size.
There exists a graph~$H_T$~containing~$T$~as an automorphism invariant induced subgraph so that the action of~$\Aut(H_T)$~is faithful on~$T$~and semiregular on the set of leaves of~$T$. Moreover,~$\Aut(H_T)$~induces the same orbits on~$T$ as~$\Aut(T)$.
\end{restatable}
Again, we prove the theorem in a constructive manner. All steps can be easily converted into an algorithm that takes as input an admissible (i.e., compatible with our necessary conditions from Section \ref{sec:necessary_conditions}) colored tree $T$ and produces a graph and cell selector with IR-tree $T$.

\section{Individualization-Refinement Trees} \label{sec:ir}
Following \cite{McKay201494} closely, we introduce the notion of an IR-tree. 
Algorithms based on the IR-paradigm explore these trees using various traversal strategies to solve graph isomorphism, graph automorphism or canonical labeling problems.

\xparagraph{Colored Graphs.} An undirected, finite graph $G = (V, E)$ consists of a set of vertices $V \subseteq \mathbb{N}$ and a set of edges $E \subseteq V^2$, where $E$ is symmetric. Set $n:=|V|$. 

The IR framework relies on coloring the vertices of a graph. 
A coloring is a surjective map $\pi \colon V \to \{1, \dots{}, k\}$. The~$i$-th \emph{cell} for $i \in \{1, \dots{}, k\}$ is~$\pi^{-1}(i) \subseteq V$. Elements in the same cell are \emph{indistinguishable}. If $|\pi| = n$, i.e., whenever each vertex has its own distinct color in $\pi$, then $\pi$ is called \emph{discrete}.
A coloring $\pi$ is \emph{finer} than $\pi'$ (and~$\pi'$ \emph{coarser} than~$\pi$) if $\pi(v) = \pi(v')$ implies~$\pi'(v) = \pi'(v')$ for all $v ,v' \in V$. 
Whenever convenient, we may also view colorings as ordered partitions instead of maps. 
A colored graph $(G, \pi)$ consists of a graph and a coloring.
 
The symmetric group on~$\{1,\ldots,n\}$ is denoted~$\Sym(n)$. An automorphism of a graph~$G$ is a bijective map $\varphi\colon V \to V$ with $G^\varphi := (\varphi(V), \varphi(E)) = (V, E) = G$. With $\Aut(G)$ we denote the automorphism group of~$G$. For a colored graph~$(G,\pi$) we require automorphisms to also preserve colors, i.e.,~$\pi(v)= \pi(\varphi(v))$ for all $v\in V$. 
We define the colored automorphism group~$\Aut(G, \pi)$ accordingly.

\xparagraph{Color Refinement and Individualization.} IR-algorithms use a procedure to heuristically refine colorings based on the degree of vertices. 
The intuition is that if two vertices have different degree, then they can not be mapped to each other by an automorphism. 
We assign vertices of different degrees distinct colors to indicate this phenomenon. 
This process is iterated using color degrees: for example, two vertices can only be mapped to each other if they have the same number of neighbors of a particular color~$i$. Therefore vertices can be distinguished according to the number of neighbors they have in color~$i$. This gives us a new, refined coloring that (potentially) distinguishes more vertices. This is repeated until the process stabilizes.

The colorings resulting from this process are called equitable colorings. A coloring~$\pi$ is \emph{equitable} if for every pair of (not necessarily distinct) colors~$i,j\in \{1,\ldots,k\}$ the number of~$j$-colored neighbors is the same for all~$i$-colored vertices. 
For a colored graph~$(G,\pi)$ there is (up to renaming of colors) a unique coarsest equitable coloring finer than~$\pi$~\cite{McKay201494}.
We denote this coloring by~$\Refx(G, \pi,  \epsilon)$, where $\epsilon$ is the empty sequence.

IR-algorithms also use \emph{individualization}.
This process artificially forces a vertex into its own cell.  
We can record which vertices have been individualized in a sequence $\nu \in V^*$.
We extend the refinement function so that~$\Refx(G, \pi,  \nu)$ is the unique coarsest equitable coloring finer than~$\pi$ in which every vertex in~$\nu$ is a singleton with its own artificial color.
Specifically, the artificial colors used to individualize $\nu$ are not interchangeable with colors introduced by the refinement itself and are ordered: the $i$-th vertex in $\nu$ is always colored using the $i$-th artificial color. 

We require this coloring to be isomorphism invariant (which means that~$\Refx(G, \pi,  \nu)(v) = \Refx(G^\varphi, \pi^\varphi,  \nu^\varphi)(v^\varphi)$ for~$\varphi\in \Sym(n)$).
There are efficient \emph{color refinement} algorithms to compute~$\Refx(G, \pi,  \nu)$, for which we refer to~\cite{McKay201494}.

We say two colored graphs $(G_1,\pi_1)$ and $(G_2,\pi_2)$ are \emph{distinguishable (by color refinement)}, if with respect to the colorings~$\Refx(G_1, \pi_1,  \epsilon)$ and $\Refx(G_2, \pi_2,  \epsilon)$ 
\begin{enumerate}
\item there is a color $c$ with differently sized cells in $G_1$ and $G_2$ (i.e., $|\Refx(G_1, \pi_1,  \epsilon)^{-1}(c)| \neq |\Refx(G_2, \pi_2,  \epsilon)^{-1}(c)|)$), 
\item or there are
vertices~$v_1\in V(G_1)$, $v_2\in V(G_2)$ of the same color~(i.e., $\Refx(G_1, \pi_1,  \epsilon)(v_1)= \Refx(G_2, \pi_2,  \epsilon)(v_2)$), such that there is a color~$c$ within which~$v_1$ and~$v_2$ have a differing number of neighbors~(i.e.,~$|\{(v_1,w)\in E(G_1)\mid \Refx(G_1, \pi_1,  \epsilon)(w)=c\}|\neq |\{(v_2,w)\in E(G_2)\mid \Refx(G_2, \pi_2,  \epsilon)(w)=c\}|$).
\end{enumerate}
Sequences (or $t$-tuples) of vertices $\nu_1\in (G_1,\pi_1)^t$ and $\nu_2\in (G_2,\pi_2)^t$ are 
\emph{distinguishable}, if the graphs $(G_1,\Refx(G_1, \pi_1, \nu_1))$ and $(G_2,\Refx(G_2, \pi_2, \nu_2))$ are.

\xparagraph{Cell Selector.} In a backtracking fashion, the goal of an IR-algorithm is to reach a discrete coloring using color refinement and individualization. For this, color refinement is first applied.
If this does not yield a discrete coloring, individualization is applied, branching over all vertices in one non-singleton cell. The task of the \emph{cell selector} is to isomorphism invariantly pick the non-singleton cell. After individualization, color refinement is applied again and the process continues recursively.
Formally, a cell selector is a function $\Sel \colon \mathcal{G} \times \Pi \to 2^V$ (where $\mathcal{G}$ denotes the set of all graphs and $\Pi$ denotes the set of all colorings), satisfying:
\begin{itemize}
	\item Isomorphism invariance, i.e., $\Sel(G^\varphi,\pi^\varphi) = \Sel(G,\pi)^\varphi$ for $\varphi \in \Sym(n)$.
	\item If $\pi$ is discrete then $\Sel(G,\pi) = \emptyset$.
	\item If $\pi$ is not discrete then  $|\Sel(G,\pi)| > 1$ and $\Sel(G,\pi)$ is a cell of $\pi$.
\end{itemize}

\xparagraph{IR-Tree.}
We describe the IR-tree $\Gamma_{\Sel}(G, \pi)$ 
of a colored graph~$(G, \pi)$, which depends on a chosen cell selector~$\Sel$.
Essentially, IR-Trees simply describe the call-trees stemming from the aforementioned backtracking procedure.
Nodes of the search tree are sequences of vertices of~$G$.
  The root of $\Gamma_{\Sel}(G, \pi)$ is the empty sequence $\epsilon$.
  If $\nu$ is a node in $\Gamma_{\Sel}(G, \pi)$ and $C = \Sel(G,\Refx(G, \pi, \nu))$, then the set of children of~$\nu$ is $\{\nu.v \; | \; v \in C \}$, i.e., all extensions of~$\nu$ by one vertex~$v$ of~$C$.

By $\Gamma_{\Sel}(G, \pi, \nu)$ we denote the subtree of $\Gamma_{\Sel}(G, \pi)$ rooted in $\nu$. We omit the index $\Sel$ when apparent from context. 

We recite the following fact on isomorphism invariance of the search tree as given in \cite{McKay201494}, which follows from the isomorphism invariance of $\Sel$ and $\Refx$:
\begin{lemma} \label{lem:auto_tree_correspondence1} If $\nu$ is a node of $\Gamma(G, \pi)$
 and $\varphi \in \Aut(G, \pi)$, then $\nu^\varphi$ is a node of $\Gamma(G, \pi)$ and $\Gamma(G, \pi, \nu)^\varphi = \Gamma(G, \pi, \nu^\varphi)$.
\end{lemma}

\xparagraph{Quotient Graph.} 
The IR-tree itself can be exponentially large in the order of~$G$~\cite{DBLP:conf/stoc/NeuenS18}. 
To decrease its size IR-algorithms use a pruning mechanism. 
For this a \emph{node invariant} is used. 
A node invariant is a function $\Inv \colon \mathcal{G} \times \Pi \times V^* \to I$ that assigns to each sequence of nodes of the tree a value in a totally ordered set $I$. It satisfies the following.
\begin{itemize}
	\item Isomorphism invariance, i.e., $\Inv(G, \pi, \nu_1) = \Inv(G^\varphi, \pi^\varphi, \nu_1^\varphi)$ for $\varphi \in \Sym(n)$.
	\item If $|\nu_1| = |\nu_2|$ and $\Inv(G, \pi, \nu_1) < \Inv(G, \pi, \nu_2)$, then for all nodes $\nu_1' \in \Gamma(G, \pi, \nu_1)$ and $\nu_2' \in \Gamma(G, \pi, \nu_2)$ it holds that $\Inv(G, \pi, \nu_1') < \Inv(G, \pi, \nu_2')$.
\end{itemize}

The particular way the node invariant can be exploited depends on the problem to be solved. When solving for graph isomorphism, the algorithm may prune all nodes with an invariant differing from an arbitrary node invariant. However, when algorithms want to compute a canonical labeling, they must find a specific canonical node invariant to continue with. However, in the context of the present work these details are not important.

Most IR-algorithms use a specific invariant, the so-called \emph{quotient graph}, which is naturally produced by color refinement.

For an equitable coloring~$\pi$ of a graph~$G$, the quotient graph $Q(G, \pi)$ captures the information of how many neighbors vertices from one cell have in another cell.
Quotient graphs are complete directed graphs in which each vertex has a self-loop.
They include vertex colors as well as edge colors. 
The vertex set of $Q(G, \pi)$ is the set of all colors of~$(G,\pi)$, i.e., $V(Q(G, \pi)) := \pi(V(G))$.
The vertices are colored with the color of the cell they represent in~$G$.
We color the edge $(c_1, c_2)$ with the number of neighbors a vertex of cell $c_1$ has in cell $c_2$ (possibly $c_1 = c_2$).
Since $\pi$ is equitable, all vertices of $c_1$ have the same number of neighbors in $c_2$. 

A crucial fact is that graphs are \emph{indistinguishable by color refinement} if and only if their quotient graphs on the coarsest equitable coloring are equal.

We should also remark that quotient graphs are indeed \emph{complete invariants}, yielding the following property.
\begin{lemma} \label{lem:complete_invariant} Let $\nu, \nu'$ be leaves of $\Gamma(G, \pi)$. 
  There exists an automorphism $\varphi \in \Aut(G, \pi)$ with $\nu = \varphi(\nu')$ if and only if $Q(G, \Refx(G, \pi, \nu)) = Q(G, \Refx(G, \pi, \nu'))$.
\end{lemma}
Consistent with the colors of trees used in \cite{anders_et_al:LIPIcs.ICALP.2021.16}, we may also view quotient graphs as a way to color IR-trees themselves, i.e., where we color a node $\nu$ with $Q(G, \Refx(G, \pi, \nu))$.

\section{Necessary Conditions for IR-Trees} \label{sec:necessary_conditions}
We collect necessary conditions for the structure of IR-trees.
Since IR-trees are the result of a branching process, they are naturally irreducible (no node has exactly one child). Also, indistinguishable leaves can be mapped to each other.
\begin{lemma} \label{lem:ir_irreducible} IR-trees are irreducible. 
\end{lemma}

\begin{lemma}\label{lem:leaf_color_complete_inv} Let $l_1, l_2$ be two leaves of an IR-tree~$(T,\pi)$. If $l_1$ and $l_2$ are indistinguishable, there is an automorphism~$\varphi\in \Aut(T, \pi)$ mapping $l_1$ to $l_2$.
\end{lemma}

\begin{lemma}[see e.g.~\cite{DBLP:conf/alenex/AndersS21}] \label{lem:leaf_auto_correspondence0} A leaf $l$ can be mapped to exactly $|\Aut(G, \pi)|$ leaves in $\Gamma(G, \pi)$ using elements of the automorphism group $\Aut(G, \pi)$.
\end{lemma}
It follows that all classes of indistinguishable leaves have equal size.

Since in color refinement, partitionings and hence quotient graphs only ever become finer and more expressive, the following properties hold.
\begin{lemma}\label{lem:quot_become_finer}
Let $n_1, n_2$ be two nodes of an IR-tree where $n_i$ is on level $l_i$.
\begin{enumerate}
  \item If $l_1 \neq l_2$, then $n_1$ and $n_2$ are distinguishable.
  \item Consider the two walks starting in the root and ending in $n_1$ and in $n_2$, respectively. If in these walks two nodes on the same level are distinguishable then $n_1$ and $n_2$ are distinguishable.
\end{enumerate}
\end{lemma}
\tikzset{
  triangle/.style={isosceles triangle,draw,shape border rotate=90, dashed, minimum height=5mm, minimum width=7.5mm, inner sep=0},
}
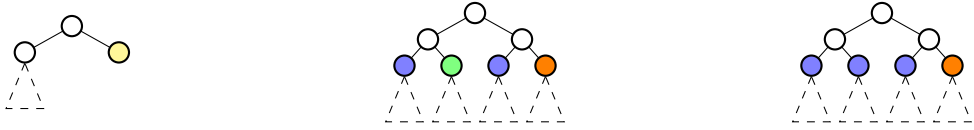
\begin{figure}[t]
  \centering
  \tikzset{
  triangle/.style={isosceles triangle,draw,shape border rotate=90, dashed, minimum height=5mm, minimum width=5mm, inner sep=0},
}
  \begin{center}
    \pgfdeclarelayer{bg} 
  \pgfsetlayers{bg,main} 
  \begin{minipage}{0.32\linewidth}
    \centering
\noindent\begin{tikzpicture}[scale=0.5,
  every node/.style = {minimum width = 0.7em, inner sep = 0, outer sep = 0, draw, circle, fill=gray!50},
  level/.style = {sibling distance = 25mm/#1, level distance=7mm}
  ]
  \node[thick, fill=white] {}
  child {	node[thick, fill=white] {} child{node[triangle,fill=white,yshift=-2.3mm] {}}
  }
  child {	node[thick, fill=yellow!50] {}
  };
\end{tikzpicture}
\end{minipage}
\begin{minipage}{0.32\linewidth}
  \centering
  \noindent\begin{tikzpicture}[scale=0.5,
    every node/.style = {minimum width = 0.7em, inner sep = 0, outer sep = 0, draw, circle, fill=gray!50},
    level/.style = {sibling distance = 25mm/#1, level distance=7mm}
    ]
    \node[thick, fill=white] {}
    child {	node[thick, fill=white] {}
        child {	node[thick, fill=lightblue] {} child{node[triangle,fill=white,yshift=-2.3mm] {}}
        }
        child {	node[thick, fill=green!50] {} child{node[triangle,fill=white,yshift=-2.3mm] {}}
        }
    }
    child {	node[thick, fill=white] {}
        child {	node[fill=lightblue, thick] {} child{node[triangle,fill=white,yshift=-2.3mm] {}}
        }
        child {	node[fill=orange, thick] {} child{node[triangle,fill=white,yshift=-2.3mm] {}}
        }
    };
  \end{tikzpicture}
  \end{minipage}
  \begin{minipage}{0.32\linewidth}
    \centering
    \noindent\begin{tikzpicture}[scale=0.5,
      every node/.style = {minimum width = 0.7em, inner sep = 0, outer sep = 0, draw, circle, fill=gray!50},
      level/.style = {sibling distance = 25mm/#1, level distance=7mm}
      ]
      \node[thick, fill=white] {}
      child {	node[thick, fill=white] {}
          child {	node[thick, fill=lightblue] {} child{node[triangle,fill=white,yshift=-2.3mm] {}}
          }
          child {	node[thick, fill=lightblue] {} child{node[triangle,fill=white,yshift=-2.3mm] {}}
          }
      }
      child {	node[thick, fill=white] {}
          child {	node[fill=lightblue, thick] {} child{node[triangle,fill=white,yshift=-2.3mm] {}}
          }
          child {	node[fill=orange, thick] {} child{node[triangle,fill=white,yshift=-2.3mm] {}}
          }
      };
    \end{tikzpicture}
    \end{minipage}
\end{center}
\caption{Forbidden structures in asymmetric binary IR-trees.}
\label{fig:forbidden_structures}
\end{figure}
Some further restrictions apply specifically in the case of cells of size $2$.
\begin{lemma}[Forbidden Binary Structures] \label{lem:forbiddenbinary}
\begin{enumerate}
\item If a node~$n$ has two children~$n_1$ and~$n_2$, then it cannot be that exactly one of the children~$n_1$ or~$n_2$ is a leaf (see Figure~\ref{fig:forbidden_structures}, left).
\item If~$n_1,n_2$ are any two nodes and~$n_1$ has exactly 2 children then the multiset of colors of the children of~$n_1$ and~$n_2$ are equal or disjoint (Figure~\ref{fig:forbidden_structures}, middle and right).
\end{enumerate}
\end{lemma}
\begin{proof} Part 1 follows from the fact that individualizing one vertex in a cell of size $2$  also individualizes the other vertex of the cell.

For Part 2 we note that individualization of a child of~$n_1$ also individualizes the other child of~$n_1$ and vice versa. This implies that
if a child~$c_2$ of~$n_2$ has the same color as some child~$c_1$ of~$n_1$, then by definition, individualization of $c_1$ and $c_2$, respectively, produces indistinguishable colorings.
So in this case there is a one-to-one correspondence between the colors of the children of~$n_1$ and those of~$n_2$.
\end{proof}
It is easy to see that if at any point the cell selector chooses differently sized cells in different branches, the branches subsequently become distinguishable.
However, if we assume cell selectors only base their decision on the quotient graph, this restriction applies earlier. 
More specifically, we call a cell selector \emph{quotient-graph-based}, whenever the result of the cell selector depends only on the quotient graph rather than other aspects of~$G$ and~$\pi$ (i.e., we have $\Sel(Q(G,\pi))$ rather than~$\Sel(G,\pi)$).
Then, we have the following.
\begin{lemma} \label{lem:cell_size_and_isotype_quotient}
If two nodes~$n$ and~$n'$ in an IR-tree are indistinguishable, then their parents have the same number of children. If additionally the cell selector is quotient-graph-based then~$n$ and~$n'$ also have the same number of children.
\end{lemma}
Restricting the cell selector to quotient graphs thus changes whether we can distinguish nodes with a differing number of children \emph{before} or \emph{after} individualizing one more vertex.
We may even distinguish cells \emph{before} individualization in both cases, if we include the decision of the cell selector into the invariant itself (i.e., using $(Q(G,\pi), \Sel(G,\pi))$ instead of $Q(G,\pi)$, which is clearly only more expressive in case the cell selector is \emph{not} quotient-graph-based). 

In the following, we assume cell selectors are indeed quotient-graph-based.
Since we only require a less powerful cell selector, our construction becomes more general.
However, in the construction, we could alternatively drop the additional restriction above with minor adjustments by allowing a more powerful cell selector.

For the remainder of this paper we say that a tree fulfills the \emph{necessary conditions}, if none of the conditions laid out by this section are violated. 

\newcommand{\gadgetCat}[5]{
\draw[ultra thick] (#1+0,#2+0) -- (#1+0,#2+3*#3) -- (#1+2*#3,#2+3*#3) -- (#1+2*#3,#2+0) -- (#1+0,#2+0);
\node[scale=1.5] at (#1+1*#3,#2+1.8*#3) (C) {\small \textbf{C}};
\node[scale=1] at (#1+1*#3,#2+0.7*#3) (n_1) {$#4,#5$};
\node[draw,circle,fill=black,scale=0.5] at (#1+2*#3, #2+3*#3*0.25) (C1#4#5) {};
\node[draw,circle,fill=black,scale=0.5] at (#1+2*#3, #2+3*#3*0.75) (C2#4#5) {};

}

\newcommand{\gadgetAsym}[4]{
  \draw[thick] (#1+0,#2+0) -- (#1+0,#2+3*#3) -- (#1+2*#3,#2+3*#3) -- (#1+2*#3,#2+0) -- (#1+0,#2+0);
  \node[scale=1.5] at (#1+1*#3,#2+1.5*#3) (A) {$A_{#4}$};
  \node[draw,circle,fill=black,scale=0.5] at (#1, #2+3*#3*0.25) (A1) {};
  \node[draw,circle,fill=black,scale=0.5] at (#1, #2+3*#3*0.75) (A2) {};
}

\newcommand{\gadgetUniUp}[4]{
  \draw[ thick] (#1+0,#2+0) -- (#1+0,#2+3*#3) -- (#1+2*#3,#2+3*#3) -- (#1+2*#3,#2+0) -- (#1+0,#2+0);
		\draw[ thick] (#1+1  *#3,#2+0.5*#3) -- (#1+1*#3,#2+2.5*#3);
		\draw[ thick] (#1+0.5*#3,#2+2  *#3) -- (#1+1.035*#3,#2+2.5*#3);
    \draw[ thick] (#1+1.5*#3,#2+2  *#3) -- (#1+0.965*#3,#2+2.5*#3);
  \node[draw,circle,fill=black,scale=0.5] at (#1+2*#3*0.33, #2+3*#3) (UDin1#4) {};
  \node[draw,circle,fill=black,scale=0.5] at (#1+2*#3*0.66, #2+3*#3) (UDin2#4) {};

  \node[draw,circle,fill=black,scale=0.5] at (#1+2*#3*0.33, #2+0*#3) (UDout1#4) {};
  \node[draw,circle,fill=black,scale=0.5] at (#1+2*#3*0.66, #2+0*#3) (UDout2#4) {};
}

\newcommand{\gadgetUniDown}[4]{
  \draw[ thick] (#1+0,#2+0) -- (#1+0,#2+3*#3) -- (#1+2*#3,#2+3*#3) -- (#1+2*#3,#2+0) -- (#1+0,#2+0);
		\draw[ thick] (#1+1  *#3,#2+3*#3-0.5*#3) -- (#1+1*#3,#2+3*#3-2.5*#3);
		\draw[ thick] (#1+0.5*#3,#2+3*#3-2  *#3) -- (#1+1.035*#3,#2+3*#3-2.5*#3);
    \draw[ thick] (#1+1.5*#3,#2+3*#3-2  *#3) -- (#1+0.965*#3,#2+3*#3-2.5*#3);

  \node[draw,circle,fill=black,scale=0.5] at (#1+2*#3*0.33, #2+3*#3) (UDin1#4) {};
  \node[draw,circle,fill=black,scale=0.5] at (#1+2*#3*0.66, #2+3*#3) (UDin2#4) {};

  \node[draw,circle,fill=black,scale=0.5] at (#1+2*#3*0.33, #2+0*#3) (UDout1#4) {};
  \node[draw,circle,fill=black,scale=0.5] at (#1+2*#3*0.66, #2+0*#3) (UDout2#4) {};
}

\newcommand{\gadgetDeadend}[4]{
  \draw[ultra thick] (#1+0,#2+0) -- (#1+0,#2+3*#3) -- (#1+2*#3,#2+3*#3) -- (#1+2*#3,#2+0) -- (#1+0,#2+0);
    
  \draw[ultra thick] (#1+1*#3,   #2+0.5*#3) -- (#1+1*#3,  #2+1*#3);
  \draw[ultra thick] (#1+1*#3,   #2+2.5*#3) -- (#1+1*#3,  #2+2*#3);
  \draw[ultra thick] (#1+0.5*#3, #2+2*#3)   -- (#1+1.5*#3,#2+2*#3);
  \draw[ultra thick] (#1+0.5*#3, #2+1*#3)   -- (#1+1.5*#3,#2+1*#3);

  \node[draw,circle,fill=black,scale=0.5] at (#1+2*#3*0.33, #2+3*#3) (UDin1#4) {};
  \node[draw,circle,fill=black,scale=0.5] at (#1+2*#3*0.66, #2+3*#3) (UDin2#4) {};

  \node[draw,circle,fill=black,scale=0.5] at (#1+2*#3*0.33, #2+0*#3) (UDout1#4) {};
  \node[draw,circle,fill=black,scale=0.5] at (#1+2*#3*0.66, #2+0*#3) (UDout2#4) {};
}

\section{Graph Constructions}\label{sec:graph:constr}
Given a colored tree $(T, \pi)$ which satisfies the necessary conditions, we construct a graph~$G(T, \pi)$ whose IR-tree is $(T, \pi)$, up to renaming of colors. Standard arguments show that it suffices to construct a colored graph~$G(T, \pi)$, from which an uncolored graph with the same IR-tree can be obtained. 
We make abundant use of gadget constructions, which we describe first. 

\subsection{Gadgets}
All our gadgets have multiple \emph{input} and \emph{output gates}. 
Each gate is a pair of vertices that together form their own color class in the gadget.
Vertices in the gates are the only vertices of the gadgets connected to other vertices outside the gadget.
We say that vertices labeled with $b_i$ denote the ``input'', while $a_i$ denote ``output''.

Gates can be \emph{activated} by which we mean the process of distinguishing the vertices of the gate pair into distinct color classes, and applying color refinement afterwards.
We say activation \emph{discretizes} the gadget if the resulting stable coloring on the gadget vertices is discrete.

We should note that three of the gadgets we are about to present (specifically the $\AND_i$, Unidirectional and Dead End gadget) have already been used in other contexts related to color refinement~\cite{DBLP:journals/mst/BerkholzBG17, DBLP:conf/mfcs/ArvindFKKR16,DBLP:conf/focs/Grohe96, DBLP:conf/icalp/AndersSW21}. 

\xparagraph{$\AND_i$ Gadget \cite{DBLP:journals/mst/BerkholzBG17, DBLP:conf/mfcs/ArvindFKKR16,DBLP:conf/focs/Grohe96}.}
The $\AND_2$ gadget as illustrated in Figure~\ref{fig:appendix:AND2} realizes the logical conjunction of gates with respect to color refinement, and an XOR gadget with respect to automorphisms.
 
Given $i>2$, we can realize an $\AND_i$ gadget with~$i$ input gates by combining multiple $\AND_2$ gadgets in a tree-like fashion. 
The $\AND_i$ gadget is constructed by attaching the first and second input gate to an $\AND_2$, whose output is connected to another $\AND_2$ together with the third input gate, and so on. 
We use colors to order the input gates, i.e., we color the $i$-th input gate with color $i$. 

We define the special case of the $\AND_1$ gadget to simply consist of a pair of vertices that functions as the input and output gate at the same time.
\begin{lemma}[\cite{DBLP:conf/focs/Grohe96}]
	The $\AND_i$ gadget admits automorphisms that flip the output gate and either one of the input gates while fixing other input gates. As long as some input gate remains unsplit, the output gate is not split
	but activating all inputs discretizes the gadget.
\end{lemma}

\xparagraph{Unidirectional and Dead End Gadget \cite{DBLP:conf/mfcs/ArvindFKKR16,DBLP:conf/focs/Grohe96,DBLP:conf/icalp/AndersSW21}.}
Next, we describe gadgets through which gate activation can be propagated or blocked depending on the direction of the gadget. 
Specifically we construct the unidirectional gadget (Figure~\ref{fig:appendix:unidirectional_gadget}) and the dead end gadget (Figure~\ref{fig:appendix:deadend_gadget}).
Note that the two gadgets are indistinguishable from each other by color refinement. 
The smaller vertices depicted in Figure \ref{fig:appendix:directional_gadgets} have been included to guarantee that the gadgets become discrete after the input and output gate has been split and can otherwise be ignored. 

\begin{figure}
  \centering
  \begin{subfigure}{0.32\linewidth}
    \centering
      \centering
      \begin{tikzpicture}[scale=0.75]
      \node[draw,circle,fill=blue,scale=0.5] at (0,0) (b_0) {};
      \node[draw,circle,fill=blue,scale=0.5] at (0,1) (b_1) {};
      \node[draw,circle,fill=green,scale=0.5] at (0,2) (b_2) {};
      \node[draw,circle,fill=green,scale=0.5] at (0,3) (b_3) {};
      
      \node[draw,circle,fill,scale=0.5] at (1,0) (c_0) {};
      \node[draw,circle,fill,scale=0.5] at (1,1) (c_1) {};
      \node[draw,circle,fill,scale=0.5] at (1,2) (c_2) {};
      \node[draw,circle,fill,scale=0.5] at (1,3) (c_3) {};
      
      \node[draw,circle,fill=red,scale=0.5] at (2,1) (a_0) {};
      \node[draw,circle,fill=red,scale=0.5] at (2,2) (a_1) {};
      
      \node[] at (-0.5,0.2) (b_0_label) {$b_0$};
      \node[] at (-0.5,1.2) (b_1_label) {$b_1$};
      \node[] at (-0.5,2.2) (b_2_label) {$b_2$};
      \node[] at (-0.5,3.2) (b_3_label) {$b_3$};
      
      
      \node[] at (2.4,1.2) (a_0_label) {$a_0$};
      \node[] at (2.4,2.2) (a_1_label) {$a_1$};
      
      \node[fill=none,stroke=none] at (0,3.5) (dummy) {};
      
      \draw (b_0) -- (c_0);
      \draw (b_0) -- (c_2);
      \draw (b_1) -- (c_1);
      \draw (b_1) -- (c_3);
      \draw (b_2) -- (c_1);
      \draw (b_2) -- (c_2);
      \draw (b_3) -- (c_0);
      \draw (b_3) -- (c_3);
      
      \draw (c_0) -- (a_0);
      \draw (c_1) -- (a_0);
      \draw (c_2) -- (a_1);
      \draw (c_3) -- (a_1);
      \end{tikzpicture}
      \caption{$\AND_2$ gadget.}
      \label{fig:appendix:AND2}
  \end{subfigure}
\begin{subfigure}{0.32\linewidth}
  \centering
  \begin{tikzpicture}[scale=0.75]
  \node[draw,circle,fill,scale=0.5] at (0,0) (b_0) {};
  \node[draw,circle,fill,scale=0.5] at (0,1) (b_1) {};
  \node[draw,circle,fill,scale=0.5] at (0,2) (b_2) {};
  \node[draw,circle,fill,scale=0.5] at (0,3) (b_3) {};
  
  \node[draw,circle,fill,scale=0.5] at (1,0) (c_0) {};
  \node[draw,circle,fill,scale=0.5] at (1,1) (c_1) {};
  \node[draw,circle,fill,scale=0.5] at (1,2) (c_2) {};
  \node[draw,circle,fill,scale=0.5] at (1,3) (c_3) {};
  
  \node[draw,circle,fill=red,scale=0.5] at (2,1) (a_0) {};
  \node[draw,circle,fill=red,scale=0.5] at (2,2) (a_1) {};
  
  \node[draw,circle,fill=blue,scale=0.5] at (-1,1) (a_0') {};
  \node[draw,circle,fill=blue,scale=0.5] at (-1,2) (a_1') {};
  
  \node[draw,circle,fill,scale=0.35] at (-0.8,2.5) (s_1') {};
  \node[draw,circle,fill,scale=0.35] at (-0.8,0.5) (s_0') {};
  \node[draw,circle,fill,scale=0.35] at (1.8,2.5) (s_1) {};
  \node[draw,circle,fill,scale=0.35] at (1.8,0.5) (s_0) {};

  \node[] at (2.2,1.2) (a_0_label) {$a_0$};
  \node[] at (2.2,2.2) (a_1_label) {$a_1$};
  
  \node[] at (-1.3,1.2) (a_0'_label) {$b_0$};
  \node[] at (-1.3,2.2) (a_1'_label) {$b_1$};

  \node[fill=none,stroke=none] at (0,3.5) (dummy) {};
  
  \draw (a_0') -- (s_0'); 
  \draw (b_0) -- (s_0');
  \draw (b_1) -- (s_0');
  \draw (a_1') -- (s_1');
  \draw (b_2) -- (s_1');
  \draw (b_3) -- (s_1');
  \draw (a_0) -- (s_0);
  \draw (c_0) -- (s_0);
  \draw (c_1) -- (s_0);
  \draw (a_1) -- (s_1);
  \draw (c_2) -- (s_1);
  \draw (c_3) -- (s_1);
   
  \draw (b_0) -- (c_0);
  \draw (b_0) -- (c_2);
  \draw (b_1) -- (c_1);
  \draw (b_1) -- (c_3);
  \draw (b_2) -- (c_1);
  \draw (b_2) -- (c_2);
  \draw (b_3) -- (c_0);
  \draw (b_3) -- (c_3);
  
  \draw (c_0) -- (a_0);
  \draw (c_1) -- (a_0);
  \draw (c_2) -- (a_1);
  \draw (c_3) -- (a_1);
  
  \draw (b_0) -- (a_0');
  \draw (b_2) -- (a_0');
  \draw (b_1) -- (a_1');
  \draw (b_3) -- (a_1');
  \end{tikzpicture}
  \caption{Unidirectional gadget.}
  \label{fig:appendix:unidirectional_gadget}
\end{subfigure}
\begin{subfigure}{0.32\linewidth}
  \centering
  \begin{tikzpicture}[scale=0.75]
  \node[draw,circle,fill,scale=0.5] at (0,0) (b_0) {};
  \node[draw,circle,fill,scale=0.5] at (0,1) (b_1) {};
  \node[draw,circle,fill,scale=0.5] at (0,2) (b_2) {};
  \node[draw,circle,fill,scale=0.5] at (0,3) (b_3) {};
  
  \node[draw,circle,fill,scale=0.5] at (1,0) (c_0) {};
  \node[draw,circle,fill,scale=0.5] at (1,1) (c_1) {};
  \node[draw,circle,fill,scale=0.5] at (1,2) (c_2) {};
  \node[draw,circle,fill,scale=0.5] at (1,3) (c_3) {};
  
  \node[draw,circle,fill=red,scale=0.5] at (2,1) (a_0) {};
  \node[draw,circle,fill=red,scale=0.5] at (2,2) (a_1) {};
  
  \node[draw,circle,fill=blue,scale=0.5] at (-1,1) (a_0') {};
  \node[draw,circle,fill=blue,scale=0.5] at (-1,2) (a_1') {};
  
  \node[draw,circle,fill,scale=0.35] at (-0.8,2.5) (s_1') {};
  \node[draw,circle,fill,scale=0.35] at (-0.8,0.5) (s_0') {};
  \node[draw,circle,fill,scale=0.35] at (1.8,2.5) (s_1) {};
  \node[draw,circle,fill,scale=0.35] at (1.8,0.5) (s_0) {};

  \node[] at (2.2,1.2) (a_0_label) {$a_0$};
  \node[] at (2.2,2.2) (a_1_label) {$a_1$};
  
  \node[] at (-1.3,1.2) (a_0'_label) {$b_0$};
  \node[] at (-1.3,2.2) (a_1'_label) {$b_1$};

  \node[fill=none,stroke=none] at (0,3.5) (dummy) {};
  
  \draw (a_0') -- (s_0'); 
  \draw (b_0) -- (s_1');
  \draw (b_1) -- (s_0');
  \draw (a_1') -- (s_1');
  \draw (b_2) -- (s_0');
  \draw (b_3) -- (s_1');
  \draw (a_0) -- (s_0);
  \draw (c_0) -- (s_1);
  \draw (c_2) -- (s_0);
  \draw (a_1) -- (s_1);
  \draw (c_1) -- (s_0);
  \draw (c_3) -- (s_1);
  
  \draw (b_0) -- (c_0);
  \draw (b_0) -- (c_2);
  \draw (b_1) -- (c_1);
  \draw (b_1) -- (c_3);
  \draw (b_2) -- (c_0);
  \draw (b_2) -- (c_2);
  \draw (b_3) -- (c_1);
  \draw (b_3) -- (c_3);
  
  \draw (c_0) -- (a_0);
  \draw (c_1) -- (a_0);
  \draw (c_2) -- (a_1);
  \draw (c_3) -- (a_1);
  
  \draw (b_0) -- (a_0');
  \draw (b_2) -- (a_1');
  \draw (b_1) -- (a_0');
  \draw (b_3) -- (a_1');
  \end{tikzpicture}
  \caption{Dead End gadget.}
  \label{fig:appendix:deadend_gadget}
\end{subfigure}
\caption{The $\AND_2$ gadget and two variants of directional gadgets.}
\label{fig:appendix:directional_gadgets}
\end{figure}
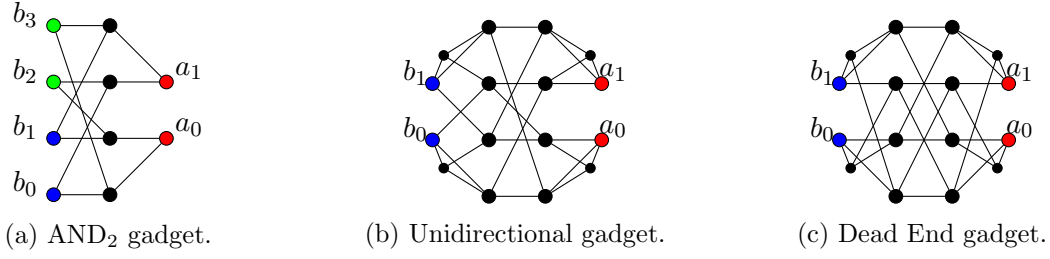

\begin{lemma}
	The unidirectional and dead end gadget are indistinguishable by color refinement.
	In the unidirectional case, activating the input discretizes the gate but activating the output does not split the input gate.
	In the dead end case, both input and output have to be activated to discretize the gadget.
\end{lemma}

\xparagraph{Asymmetry Gadgets.}
Our next gadgets only have one gate (see Figure~\ref{fig:appendix:onegategadgets}).
Both of the asymmetry gadgets $A_1$ and $A_2$ (Figures~\ref{fig:appendix:asymmetry_gadget_t1} and~\ref{fig:appendix:asymmetry_gadget_t2}) have the crucial property that the two gate vertices of either gadget are initially indistinguishable by color refinement, but
individualizing one of the gate vertices leads to a different quotient graph than individualizing the other gate vertex.
\begin{lemma}
	The asymmetry gadgets form asymmetric graphs that are stable under color refinement.
  Activating the input gate discretizes the gadget and we obtain two non-isomorphic colorings depending on which vertex was individualized.
  Furthermore, $A_1 \ncong A_2$.
\end{lemma}

\begin{figure}
\centering
\begin{subfigure}{0.49\linewidth}
  \centering
  \begin{tikzpicture}[scale=0.75]
\draw[draw=none, use as bounding box](-1,-1.875) rectangle (4,1.875);
		\node[draw,circle,fill=green,scale=0.5] at (-1,0) (a_1) {};
		\node[draw,circle,fill=green,scale=0.5] at (4,0) (a_2) {};
		
		\node[draw,circle,fill,scale=0.5] at (1.5,0.8) (l_5) {};
		\node[draw,circle,fill,scale=0.5] at (1.5,0) (l_6) {};
		\node[draw,circle,fill,scale=0.5] at (1,0.8) (l_2) {};
		\node[draw,circle,fill,scale=0.5] at (1,0.2) (l_3) {};
		\node[draw,circle,fill,scale=0.5] at (1,-0.8) (l_4) {};
		\node[draw,circle,fill,scale=0.5] at (0.5,-0.2) (l_1) {};
		
		\node[draw,circle,fill,scale=0.5] at (2.3,-0.3) (r_6) {};
		\node[draw,circle,fill,scale=0.5] at (2.3,0.3) (r_5) {};
		\node[draw,circle,fill,scale=0.5] at (1.7,-0.9) (r_4) {};
		\node[draw,circle,fill,scale=0.5] at (1.8,-0.2) (r_3) {};
		\node[draw,circle,fill,scale=0.5] at (1.8,0.2) (r_2) {};
		\node[draw,circle,fill,scale=0.5] at (1.5,-0.6) (r_1) {};
		
		\node[] at (0.25,-1.25) (l_label) {$\mathcal{F}$};

		\draw (0,-1.5) -- (0,1.5);
		\draw (0,1.5) -- (3,1.5);		
		\draw (3,1.5) -- (3,-1.5);
		\draw (3,-1.5) -- (0,-1.5);
		
		\draw (a_1) -- (l_1);
		\draw (a_1) -- (l_2);
		\draw (a_1) -- (l_3);
		\draw (a_1) -- (l_4);
		\draw (a_1) -- (l_5);
		\draw (a_1) -- (l_6);
		
		\draw (a_2) -- (r_1);
		\draw (a_2) -- (r_2);
		\draw (a_2) -- (r_3);
		\draw (a_2) -- (r_4);
		\draw (a_2) -- (r_5);
		\draw (a_2) -- (r_6);
		
		\draw (l_1) -- (l_2);
		\draw (l_1) -- (l_3);
		\draw (l_1) -- (l_4);
		\draw (l_2) -- (l_3);
		\draw (l_2) -- (l_5);
		\draw (l_3) -- (l_6);
		\draw (l_4) -- (r_1);
		\draw (l_4) -- (r_4);
		\draw (l_5) -- (r_2);
		\draw (l_5) -- (r_5);
		\draw (l_6) -- (r_2);
		\draw (l_6) -- (r_3);
		\draw (r_1) -- (r_3);
		\draw (r_1) -- (r_4);
		\draw (r_2) -- (r_5);
		\draw (r_3) -- (r_6);
		\draw (r_4) -- (r_6);
		\draw (r_5) -- (r_6);
	\end{tikzpicture}
	\caption{The asymmetry gadget $A_1$.}
    \label{fig:appendix:asymmetry_gadget_t1}
  \end{subfigure}
  \begin{subfigure}{0.49\linewidth}
    \centering
    \begin{tikzpicture}[scale=0.75]
  \draw[draw=none, use as bounding box](-1,-1.875) rectangle (4,1.875);
      \node[draw,circle,fill=green,scale=0.5] at (-1,0) (a_1) {};
      \node[draw,circle,fill=green,scale=0.5] at (4,0) (a_2) {};
      
      \node[draw,circle,fill,scale=0.5] at (1.5,0.8) (l_5) {};
      \node[draw,circle,fill,scale=0.5] at (1.5,0) (l_6) {};
      \node[draw,circle,fill,scale=0.5] at (1,0.8) (l_2) {};
      \node[draw,circle,fill,scale=0.5] at (1,0.2) (l_3) {};
      \node[draw,circle,fill,scale=0.5] at (1,-0.8) (l_4) {};
      \node[draw,circle,fill,scale=0.5] at (0.5,-0.2) (l_1) {};
      
      \node[draw,circle,fill,scale=0.5] at (2.3,-0.3) (r_6) {};
      \node[draw,circle,fill,scale=0.5] at (2.3,0.3) (r_5) {};
      \node[draw,circle,fill,scale=0.5] at (1.7,-0.9) (r_4) {};
      \node[draw,circle,fill,scale=0.5] at (1.8,-0.2) (r_3) {};
      \node[draw,circle,fill,scale=0.5] at (1.8,0.2) (r_2) {};
      \node[draw,circle,fill,scale=0.5] at (1.5,-0.6) (r_1) {};
      
      \node[] at (0.25,-1.25) (l_label) {$\mathcal{F}$};
  
      \draw (0,-1.5) -- (0,1.5);
      \draw (0,1.5) -- (3,1.5);		
      \draw (3,1.5) -- (3,-1.5);
      \draw (3,-1.5) -- (0,-1.5);
      
      \draw (a_1) -- (l_1);
      \draw (a_1) -- (l_2);
      \draw (a_1) -- (l_3);
      \draw (a_1) -- (l_4);
      \draw (a_1) -- (l_5);

      \draw (a_2) -- (l_6);
      \draw (a_1) -- (r_1);

      \draw (a_2) -- (r_2);
      \draw (a_2) -- (r_3);
      \draw (a_2) -- (r_4);
      \draw (a_2) -- (r_5);
      \draw (a_2) -- (r_6);
      
      \draw (l_1) -- (l_2);
      \draw (l_1) -- (l_3);
      \draw (l_1) -- (l_4);
      \draw (l_2) -- (l_3);
      \draw (l_2) -- (l_5);
      \draw (l_3) -- (l_6);
      \draw (l_4) -- (r_1);
      \draw (l_4) -- (r_4);
      \draw (l_5) -- (r_2);
      \draw (l_5) -- (r_5);
      \draw (l_6) -- (r_2);
      \draw (l_6) -- (r_3);
      \draw (r_1) -- (r_3);
      \draw (r_1) -- (r_4);
      \draw (r_2) -- (r_5);
      \draw (r_3) -- (r_6);
      \draw (r_4) -- (r_6);
      \draw (r_5) -- (r_6);
    \end{tikzpicture}
    \caption{The asymmetry gadget $A_2$.}
      \label{fig:appendix:asymmetry_gadget_t2}
    \end{subfigure}
	\caption{Non-isomorphic asymmetry gadgets. The two input vertices are connected regularly to disjoint halves of the Frucht graph $\mathcal{F}$.}\label{fig:appendix:onegategadgets}
\end{figure}
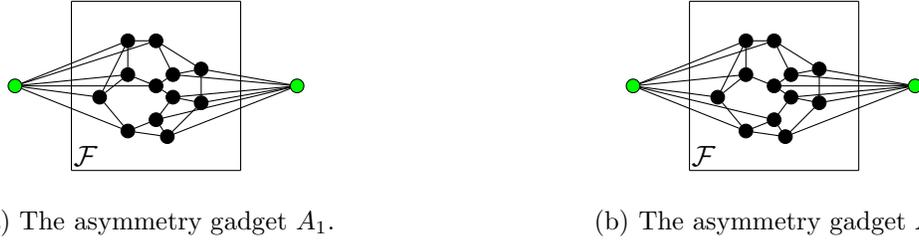

\xparagraph{Concealed Edges.}
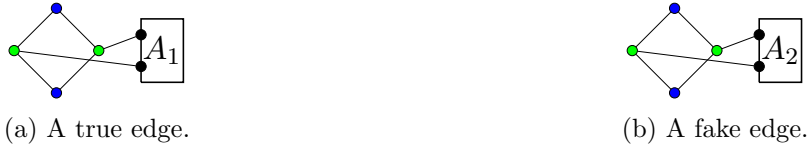
\begin{figure}
  \begin{center}
  \begin{subfigure}{0.49\linewidth}
    \centering
  \scalebox{0.75}{
  \begin{tikzpicture}[scale=0.75]
    \node[draw,circle,fill=blue,scale=0.5] at (0, 0) (i1) {};
    \node[draw,circle,fill=blue,scale=0.5] at (0, 2) (i2) {};
    \node[draw,circle,fill=green,scale=0.5] at (-1,1) (in1) {};
    \node[draw,circle,fill=green,scale=0.5] at (1, 1) (in2) {};
    \draw (i1) -- (in1);
    \draw (i2) -- (in1);
    \draw (i1) -- (in2);
    \draw (i2) -- (in2);
    \gadgetAsym{3-2*0.5}{0.25}{0.5}{1};
    \draw (A1) -- (in1);
    \draw (A2) -- (in2);
  \end{tikzpicture}}
  \caption{A true edge.}
  \end{subfigure}
  \begin{subfigure}{0.49\linewidth}
    \centering
    \scalebox{0.75}{
    \begin{tikzpicture}[scale=0.75]
      \node[draw,circle,fill=blue,scale=0.5] at (0, 0) (i1) {};
      \node[draw,circle,fill=blue,scale=0.5] at (0, 2) (i2) {};
      \node[draw,circle,fill=green,scale=0.5] at (-1,1) (in1) {};
      \node[draw,circle,fill=green,scale=0.5] at (1, 1) (in2) {};
      \draw (i1) -- (in1);
      \draw (i2) -- (in1);
      \draw (i1) -- (in2);
      \draw (i2) -- (in2);
      \gadgetAsym{3-2*0.5}{0.25}{0.5}{2};
      \draw (A1) -- (in1);
      \draw (A2) -- (in2);
    \end{tikzpicture}}
    \caption{A fake edge.} 
    \end{subfigure}
  \end{center}
  \caption{The two types of concealed edge gadgets.} \label{fig:appendix:edgegadgets}
  \end{figure}
Lastly, we describe the \emph{concealed edge gadget} that is used to hide edges from color refinement.
The gadget has two vertices that represent the endpoints of an edge (the blue vertices in Figure~\ref{fig:appendix:edgegadgets}).
The idea is that instead of an edge connecting the two vertices, we insert a concealed edge gadget. For this the gadget has a pair consisting of two \emph{inner vertices} (the green vertices in Figure~\ref{fig:appendix:edgegadgets}), which are both connected to each input vertex.
This pair is then connected to an asymmetry gadget.
We define two classes of edges, where one type of edge attaches the asymmetry gadget $A_1$ and the other $A_2$.
We call edges with asymmetry type $A_1$ \emph{true edges}, and those with $A_2$ \emph{fake edges}.

The crucial property is that as long as inner vertices of the gadgets are not distinguished, color refinement can not distinguish between true edges and fake edges.
However, if we distinguish the inner vertices, true edges can indeed be distinguished from fake edges.

We always employ this gadget within the following design pattern. 
Whenever we want to connect two sets of vertices $V_1$ and $V_2$ with edges $E \subseteq V_1 \times V_2$ in a concealed manner, we first add a concealed edge gadget between \emph{all} pairs $(v_1, v_2) \in V_1 \times V_2$.
However, only if $(v_1, v_2) \in E$, we use a \emph{true edge}, and whenever $(v_1, v_2) \notin E$ we use a \emph{fake edge}.
Finally, we connect all pairs of inner vertices of the concealed edge gadgets to some construction that is used to \emph{reveal} the edges.

The asymmetry gadget prohibits automorphisms from flipping the concealed edge gadget itself. 
However, care has to be taken when connecting the inner vertices to other constructions: it is imperative to connect the inner vertices of multiple concealed edge gadgets that are on the, say, left side of the asymmetry gadget, in the same manner. 
Otherwise, once revealed, edges could possibly be distinguished into even more categories than just fake and true edges.

\subsection{A construction for asymmetric trees}
For our construction, we first restrict ourselves to asymmetric trees, i.e., all leaves have different colors. Building on this, the following section takes symmetries into account. 
Let $(T, \pi)$ be an \emph{asymmetric}, colored tree that satisfies the necessary conditions (see Section~\ref{sec:necessary_conditions}).

We describe a graph $G(T, \pi)$ and a cell selector $S(T, \pi)$ such that $(T, \pi)$ is (up to renaming of colors) the IR-tree $\Gamma_{S(T, \pi)}(G(T, \pi))$.
We describe the construction step by step. Initially, $G(T, \pi)$ is the empty graph and we successively add more and more vertices. 

The goal is to model the graph and cell selector in such a way that there is a one-to-one correspondence between paths in $T$ and sequences of individualizations in $G(T, \pi)$. Note that such sequences are precisely the paths in the IR-tree $\Gamma_{S(T, \pi)}(G(T, \pi))$.
To guarantee such a correspondence, certain properties of the paths in the tree $T$ must translate into specific properties for their corresponding sequence of individualizations. When modeling $G(T, \pi)$ we must in particular ensure the following.
\begin{enumerate}
\item Two paths must end in nodes of different color exactly if the corresponding sequences of individualizations result in different quotient graphs.
\item A path must end in a leaf exactly if the corresponding sequence of individualizations (when followed by color refinement) results in a discrete coloring.
\end{enumerate}
These two effects are guaranteed by different parts of our construction. 
We start by describing the part of the graph on which the cell selector operates, i.e., within which cells are chosen.

\xparagraph{Selector Tree.}
One of the central difficulties  is that color refinement executed on~$T$ may actually result in a coloring that is finer than~$\pi$. This is precisely the reason why the tree must be concealed and why we cannot simply use the tree~$T$ itself.
Therefore, structural and color information about $T$ is encoded into the selector tree so that it is initially hidden from color refinement. 
In particular, the selector tree will be stable under color refinement and only after individualizations are applied, parts of the structure of $T$ are revealed. 

To construct the selector tree, we first copy all the nodes of $T$ and color each node with its level. 
To make cells appear uniform, we encode the edges of $T$ in the selector tree using concealed edges, as follows.
We fully connect nodes of level $i$ to nodes of level $i + 1$ using concealed edges, creating a complete bipartite graph. 
Only if a node $v$ at level $i + 1$ is a child of node $p$ at level $i$ in $T$, we use a true edge between $v$ and $p$. Otherwise we use a fake edge.
This guarantees that our copy of $T$ is stable under color refinement.
See Figure~\ref{fig:appendix:concealedleveledges} for an illustration.

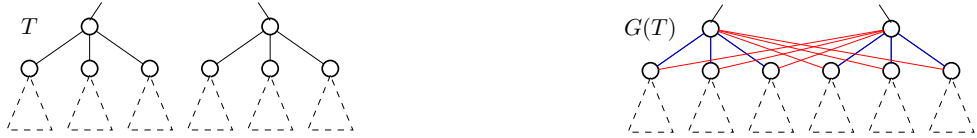
\begin{figure}
  \centering
  \begin{subfigure}{0.49\linewidth}
    \centering
    \scalebox{0.8}{\begin{tikzpicture}[
  every node/.style = {minimum width = 0.7em, inner sep = 0, outer sep = 0, draw, circle, fill=gray!50},
  level/.style = {sibling distance = 10mm/#1, level distance=7mm}
  ]
  	\node[thick, fill=white] {}
  	child {	node[thick, fill=white] {} child{node[triangle,fill=white,yshift=-0.7mm] {}}
  	}
  	child {	node[thick, fill=white] {} child{node[triangle,fill=white,yshift=-0.7mm] {}}
  	}
  	child {	 node[thick, fill=white] {} child{node[triangle,fill=white,yshift=-0.7mm] {}}
  	};
  	\node[thick, fill=white,xshift=3cm] {}
  	child {	node[thick, fill=white] {} child{node[triangle,fill=white,yshift=-0.7mm] {}}
  	}
  	child {	node[thick, fill=white] {} child{node[triangle,fill=white,yshift=-0.7mm] {}}
  	}
  	child {	 node[thick, fill=white] {} child{node[triangle,fill=white,yshift=-0.7mm] {}}
  	};
  	\node[draw=none,fill=white] at (-1,0) (T_label) {$T$};
  	\draw (0,0.1) -- (0.2,0.4); 
  	\draw[xshift=3cm] (0,0.1) -- (-0.2,0.4);
  \end{tikzpicture}}
  \end{subfigure}
  \begin{subfigure}{0.49\linewidth}
  \centering
  \scalebox{0.8}{\begin{tikzpicture}[
  every node/.style = {minimum width = 0.7em, inner sep = 0, outer sep = 0, draw, circle, fill=gray!50},
  level/.style = {sibling distance = 10mm/#1, level distance=7mm}
  ]
  	\node[thick, fill=white] (n1) {}
  	child {	node[thick, fill=white] (n2) {} child{node[triangle,fill=white,yshift=-0.7mm]  {}}
  	}
  	child {	node[thick, fill=white] (n3) {} child{node[triangle,fill=white,yshift=-0.7mm] {}}
  	}
  	child {	 node[thick, fill=white] (n4) {} child{node[triangle,fill=white,yshift=-0.7mm] {}}
  	};
  	\node[thick, fill=white,xshift=3cm] (n5) {}
  	child {	node[thick, fill=white] (n6) {} child{node[triangle,fill=white,yshift=-0.7mm] {}}
  	}
  	child {	node[thick, fill=white] (n7) {} child{node[triangle,fill=white,yshift=-0.7mm] {}}
  	}
  	child {	 node[thick, fill=white] (n8) {} child{node[triangle,fill=white,yshift=-0.7mm] {}}
  	};
  	
  	\node[draw=none,fill=white] at (-1,0) (T_label) {$G(T)$};
  	\draw (0,0.1) -- (0.2,0.4); 
  	\draw[xshift=3cm] (0,0.1) -- (-0.2,0.4);
  	\draw[blue] (n1) -- (n2);
  	\draw[blue] (n1) -- (n3);
  	\draw[blue] (n1) -- (n4);
  	\draw[blue] (n5) -- (n6);
  	\draw[blue] (n5) -- (n7);
  	\draw[blue] (n5) -- (n8);
  	\draw[red] (n1) -- (n6);
  	\draw[red] (n1) -- (n7);
  	\draw[red] (n1) -- (n8);
  	\draw[red] (n5) -- (n2);
  	\draw[red] (n5) -- (n3);
  	\draw[red] (n5) -- (n4);
  \end{tikzpicture}}
\end{subfigure}
  \caption{Connecting levels of the selector tree. Blue/red edges on the right symbolize true/fake edge gadgets} \label{fig:appendix:concealedleveledges}
\end{figure}

At some point, we will need to add another gadget construction to ensure that edges between the levels are actually revealed at the right time.
Assuming this for now, the cell selector $S(T, \pi)$ always chooses as next cell the cell that consists of the children of the node chosen last.
Here children means children with respect to true edges in the selector tree. 

\xparagraph{Colors.}
Next, we translate the colors $\pi$ of $T$ into a construction that is part of $G(T, \pi)$. 
Recall that the colors indicate whether a sequence of individualizations should lead to differing quotient graphs.
We make use of fake edges again to encode this: intuitively, we encode a one-to-one correspondence between selector tree nodes and their color in $\pi$ using concealed edges. Since the edges are concealed, they are hidden from color refinement until revealed.
\begin{figure}
  \centering
  \begin{subfigure}{0.49\linewidth}
    \centering
    \scalebox{0.8}{\begin{tikzpicture}[
    every node/.style = {minimum width = 0.7em, inner sep = 0, outer sep = 0, draw, circle, fill=gray!50},
    level/.style = {sibling distance = 10mm/#1, level distance=7mm}
    ]
      \node[thick, fill=blue] (n1) {}
      child {	node[thick, fill=green] (n2) {} child{node[triangle,fill=white,yshift=-0.7mm] {}}
      }
      child {	node[thick, fill=green] (n3) {} child{node[triangle,fill=white,yshift=-0.7mm] {}}
      }
      child {	 node[thick, fill=red] (n4) {} child{node[triangle,fill=white,yshift=-0.7mm] {}}
      };
      \node[thick, fill=blue,xshift=3cm] (n5) {}
      child {	node[thick, fill=green] (n6) {} child{node[triangle,fill=white,yshift=-0.7mm] {}}
      }
      child {	node[thick, fill=pink] (n7) {} child{node[triangle,fill=white,yshift=-0.7mm] {}}
      }
      child {	node[thick, fill=pink] (n8) {} child{node[triangle,fill=white,yshift=-0.7mm] {}}
      };
      \node[draw=none,fill=white] at (-1,0) (T_label) {$T$};
      \draw (0,0.1) -- (0.2,0.4); 
      \draw[xshift=3cm] (0,0.1) -- (-0.2,0.4);
    \end{tikzpicture}}	
    \vspace{0.75cm}
  \end{subfigure}
  \begin{subfigure}{0.49\linewidth}
    \centering
    \scalebox{0.8}{
    \begin{tikzpicture}[
    every node/.style = {minimum width = 0.7em, inner sep = 0, outer sep = 0, draw, circle, fill=gray!50},
    level/.style = {sibling distance = 10mm/#1, level distance=7mm}
    ]
      \node[thick, fill=white] (n1) {}
      child {	node[thick, fill=white] (n2) {} child{node[triangle,fill=white,yshift=-0.7mm] {}}
      }
      child {	node[thick, fill=white] (n3) {} child{node[triangle,fill=white,yshift=-0.7mm] {}}
      }
      child {	 node[thick, fill=white] (n4) {} child{node[triangle,fill=white,yshift=-0.7mm] {}}
      };
      \node[thick, fill=white,xshift=3cm] (n5) {}
      child {	node[thick, fill=white] (n6) {} child{node[triangle,fill=white,yshift=-0.7mm] {}}
      }
      child {	node[thick, fill=white] (n7) {} child{node[triangle,fill=white,yshift=-0.7mm] {}}
      }
      child {	 node[thick, fill=white] (n8) {} child{node[triangle,fill=white,yshift=-0.7mm] {}}
      };
      \node[draw=none,fill=white] at (-1,0) (T_label) {$G(T)$};
      \draw (0,0.1) -- (0.2,0.4); 
      \draw[xshift=3cm] (0,0.1) -- (-0.2,0.4);
      
      \node[circle,fill=red,thick] (r) at (0,-2.5) {};
      \node[circle,fill=blue,thick] (b) at (1,-2.5) {};
      \node[circle,fill=green,thick] (g) at (2,-2.5) {};
      \node[circle,fill=pink,thick] (p) at (3,-2.5) {};
      
      \draw (n1) -- (b);
      \draw (n5) -- (b);
      \draw (n2) -- (g);
      \draw (n3) -- (g);
      \draw (n6) -- (g);
      \draw (n4) -- (r);
      \draw (n7) -- (p);
      \draw (n8) -- (p);
    \end{tikzpicture}}	 
  \end{subfigure}
    \caption{Colors of $T$ are represented in $G(T)$ through concealed edges to special \emph{color nodes}}
    \label{fig:appendix:colornodes}
  \end{figure}
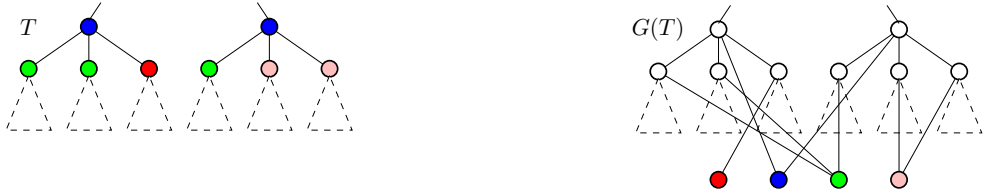
We proceed level-wise.
Let $l$ be the level under consideration. 
Let $C$ be the set of colors that appear at level $l$ of $(T, \pi)$.
For all $c \in C$, we create a unique \emph{color node} $c$ in $G(T, \pi)$. This node is also colored with $c$. 
We now connect every node at level $l$ of the selector tree to every node in $C$ using concealed edges:
we use a true edge for all pairs $(n, c)$ where $\pi(n) = c$. All other edges are fake.
See Figure~\ref{fig:appendix:colornodes} for an illustration.

As before, we still have to explain how and when edges are revealed. The idea is to always reveal the type of those concealed edges incident with node $n$ at the point in time when node $n$ is individualized.

\xparagraph{Leaf Detection.}
Whenever we individualize a node that corresponds to a leaf in~$T$, the graph $G(T, \pi)$ is supposed to become discrete, thereby terminating the IR-process. 
The first step towards this is to add a construction that detects whether a specific node $n$ in a cell was individualized.
Then, a decision can be made as to whether $n$ corresponds to a leaf or not.
\begin{figure}
  \centering
  \noindent\scalebox{0.66}{\begin{tikzpicture}[
    every node/.style = {draw,circle,fill=black}]
    \node[rectangle,draw=black,fill=white,minimum width = 10em, minimum height = 2em, thick] at (2.5,2) {\scalebox{1.1}{$\AND_3$}};
    \foreach \i in {1,...,4}{
      \ifthenelse{\i=1}{
       \node[scale=0.75,fill=orange] at (\i, 0.75) (c\i) {};
      }{
        \node[scale=0.75,fill=lightblue] at (\i, 0.75) (c\i) {};
      }
     }

     \foreach \i in {1,...,3}{
      \foreach \j in {1,...,2}{
       \node[scale=0.5] at (0.125 + \i + \j*0.25, 1.675) (i\i\j) {};
      }
     }

     \foreach \j in {1,...,2}{
       \node[scale=0.5] at (0.125 + 2 + \j*0.25, 2.325) (o\j) {};
       \node[scale=0.5] at (0.125 + 2 + \j*0.25, 3) (io\j) {};
      }

     \gadgetUniUp{3}{4.5-0.5}{0.4}{1}
     \gadgetUniDown{1}{4.5-0.5}{0.4}{2}

     \foreach \i in {1,...,2}{
      \draw[] (o\i) -- (io\i);
      \draw[] (UDout\i1) -- (io\i);
      \draw[] (UDout\i2) -- (io\i);
     }

     \foreach \i in {1,...,3}{
      \pgfmathtruncatemacro\y{\i + 1}
      \draw[] (c1) -- (i\i1);
      \draw[] (c\y) -- (i\i2);
     }

     \node[fill=none,draw=none] at (0,  5.75-1) {\scalebox{1.25}{receiver}};
     \node[fill=none,draw=none] at (4.9,5.75-1) { \scalebox{1.25}{broadcast}};
     \node[fill=none,draw=none] at (4.9+0.5,3-0.125) {\scalebox{1.25}{individualization output}};
  \end{tikzpicture}}
  \caption{Leaf detection mechanism for the leftmost node of the cell. 
  If the leftmost vertex is individualized, the $\AND_3$ gadget is activated.
  The figure only shows true edges. In the overall construction, all remaining connections between vertices of cells at level $l$ and $\AND_i$ gadgets of level $l$ are fake edges.} \label{fig:appendix:leafdetect}
\end{figure}
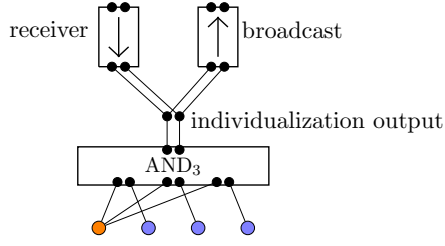
Let $s\geq2$ be the size of the current cell (in the tree~$T$ the current cell is always the set of children of some node).
For each vertex $n$ in the cell, consider all $s-1$ (unordered) pairs with other vertices of the cell.
We add an $\AND_{s-1}$ gadget and connect the left vertex of every input pair to $n$, and the other to one of the $s-1$ other vertices.
An $\AND_{s-1}$ gadget is not symmetric in its input gates, so in order to keep things symmetrical, we actually add $(s-1)!$ many $\AND_{s-1}$ gadgets for every possible order of vertices in the input.
We connect the output gates of all the $\AND_{s-1}$ gadgets to a new pair of vertices, which we call the \emph{individualization output} of $n$. 

\begin{fact}
The individualization output is activated (i.e.,~split)  whenever~$n$ is individualized.
If the cell size is larger than 2 then  the individualization output is not activated when another vertex in the cell is individualized.
\end{fact}

We should discuss the case of a size $2$ cell, in which actually both vertices of the cell become singletons when one of them is individualized.
The necessary conditions for $T$ imply that either both vertices are leaves or both vertices are internal nodes in $T$ (see Lemma~\ref{lem:forbiddenbinary}). 
Hence, while this activates the construction for both vertices, the construction is still able to model any case that satisfies the necessary conditions.

We need to ensure the construction is stable under color refinement.
Again, we can do so using concealed edges.
Consider each level $i$ in the selector tree:
all of the aforementioned edges connecting vertices of level~$i$ in the selector tree with $\AND_{s-1}$ gadgets become true edges. We then insert fake edges between nodes of the selector tree of level $i$ and the other $\AND_{s-1}$ gadgets of level $i$ if there is no true edge.
This way, the construction becomes stable under color refinement.

Whenever a node does indeed correspond to a leaf and its individualization output is activated, we want to propagate discretization to the entire graph.
We add some control structures for every node $n$ in the selector tree for this purpose.
We add a unidirectional gadget if the node is a leaf in $T$, or a dead end gadget if not.
We call this gadget the \emph{broadcast} gadget of node $n$.   
We also add a \emph{receiver} gadget to every node $n$, which is always a unidirectional gadget.

We connect the input of the broadcast gadget to the individualization output of $n$, as well as the output of the receiver gadget to the individualization output of $n$.
Next, we connect the output of the broadcast gadget to the input of \emph{all} receiver gadgets in the graph.
See Figure~\ref{fig:appendix:leafdetect} for an overview of the construction.

\begin{fact}
	When a leaf is individualized, in turn all individualization outputs in $G(T,\pi)$ are split.
	As long as no leaf is individualized, individualization outputs are split only if they belong to individualized nodes.
\end{fact}
The idea goes as follows: 
if $n$ is individualized, the individualization output is split.
If $n$ is a leaf, we want to propagate this split to all other individualization outputs, causing a discretization of the graph.
For this the broadcast gadget is activated, which sends the split to all the receiver gadgets, which in turn split their respective individualization output. If~$n$ is not a leaf, the broadcast gadget is a dead end gadget and activation of the individualization output does not have this effect.
Below, we explain how we can use the same process to reveal cells of the entire selector tree as well as actual color nodes.

\xparagraph{Revealing Cells and Colors.}
Recall that the cell selector makes choices along the selector tree and so choosing a particular cell corresponds to individualization of its parent node in the parent cell.
Assume we are individualizing a node at level $i$ of the selector tree.
At this point, we want the connections in the selector tree from level $i$ to level $i + 1$ to be revealed.
This is realized via the $\AND$-gadget construction from the previous paragraph. 
We re-use the individualization output at level $i$ to reveal the edges of the selector tree to level $i + 1$.
For this, we connect the output through a unidirectional gadget with the internal nodes of the concealed edges between level $i$ and level $i + 1$.
To be precise, for every node~$n$, we add a unidirectional gadget, the output of which is then connected to all internal nodes of the concealed edges. 
The use of unidirectional gadgets ensures that revealing the edges does not split an individualization output in the opposite direction.

Initially, the construction is stable under color refinement. 
Upon activating the unidirectional gadget, i.e., after a node on the previous level has been individualized, all true edges are distinguishable from fake edges.
Hence, actual connections to cells are visible to color refinement.

\begin{fact}
	When a node at level $i$ is individualized its color and its edges to level $i+1$ are revealed.
	Before individualizing a node at level $i+1$, these are the only revealed edges connected to level $i+1$.  
\end{fact}

In order to actually activate the individualization output, we also need to reveal edges from level~$i+1$ nodes to the $\AND_{s-1}$ gadgets. 
Hence, we do the same construction as above, connecting the unidirectional gadgets we added on level $i$ to reveal these edges on level $i + 1$. 

Note that the construction guarantees that if two nodes $n_1, n_2$ at level $l$ of $T$ have a different number of children, then $n_1$ and $n_2$ are distinguished. 
This reflects the necessary requirement discussed in Lemma~\ref{lem:cell_size_and_isotype_quotient}.
As mentioned there, this restriction could be avoided through the use of a more powerful cell selector.

For the very first level of the selector tree, the immediate children of the root, we remove the concealed edge construction by removing fake edges, such that the level is initially revealed.

Finally, the same technique is also used to reveal colors. 
We connect the individualization output of node $n$ at level $l$ to the inner vertices of the concealed edges between $n$ and the color nodes $C$ of level $l$. 
This immediately reveals the color of $n$ whenever we individualize $n$. 
In this case, we need no special construction for level $1$.

\subsection{Generating symmetries}

We expand our construction so that it can also handle colored trees $(T,\pi)$ with prescribed symmetries.
As such, the graph $G(T,\pi)$ can also be built from a tree $(T,\pi)$ that is not necessarily asymmetric. 
In this case, sequences of individualizations along root-to-leaf paths still produce the desired tree $(T,\pi)$ as a subtree of $\Gamma_{S(T,\pi)}(G(T,\pi))$. 
However, $G(T,\pi)$ is supposed to become discrete after the IR-process reaches a leaf of $(T,\pi)$, but at this point the selector tree in $G(T,\pi)$ is only split up to orbits that correspond to orbits of $T$.

\begin{figure}[t]
  \centering
  \noindent\begin{tikzpicture}[
    every node/.style = {minimum width = 0.7em, inner sep = 0, outer sep = 0, draw, circle, fill=white, scale=0.75, thick},
    level 1/.style = {sibling distance = 38mm, level distance=3mm},
    level 2/.style = {sibling distance = 15mm, level distance=3mm},
    level 3/.style = {sibling distance = 5mm, level distance=3mm},
    hidee/.style={edge from parent/.style={draw=none}},
    norme/.style={edge from parent/.style={black,thin,draw}}
    ]
    \node[] (root) {}
    child[hidee] {	node[draw=none] {}
        child[hidee]{ node (leftbranch) {}
          child[norme] {	node[fill=orange] (O1) {} 
          }
          child[norme] {	node[fill=orange] (O2) {} 
          }
          child[norme] {	node[fill=orange] (O3) {} 
          }
          child[norme] {	node[fill=orange] (O4) {} 
          }
          child[norme] {	node[fill=orange] (O5) {} 
          }
          child[norme] {	node[fill=orange] (O6) {} 
          }
        }
    }
    child {	node[thick] {}
        child {	node [] {} child {	node[fill=lightblue] (B1) {}} child {	node[fill=lightblue] (B2) {}} child {	node[fill=lightblue] (B3) {}}
        }
        child {	node[] {} child {	node[fill=lightblue] (B4) {}} child {	node[fill=lightblue] (B5) {}} child {	node[fill=lightblue] (B6) {}}
        }
    }
    child {	node[thick,draw=black] {}
        child {	node[] {} child {	node[fill=green!50] (Y1) {}} child {	node[fill=green!50] (Y2) {}} 
        }
        child {	node[] {} child {	node[fill=green!50] (Y3) {}} child {	node[fill=green!50] (Y4) {}} 
        }
        child {	node[] {} child {	node[fill=green!50] (Y5) {}} child {	node[fill=green!50] (Y6) {}} 
        }
    };

    \draw (root) -- (leftbranch);

    \foreach \i in {1,...,2}{
       \node[fill=gray!33] at (-2*0.5 + 0.25 + 0.5 * \i - 1.5, -2.5) (c1\i) {};
     }
    
     \foreach \i in {1,...,2}{
      \pgfmathtruncatemacro\y{Mod(\i,2)+1}
      \ifthenelse{\i=2}{
        \draw[thick,-latex] (c1\i) to [bend left=90,looseness=1.9] (c1\y);
      }{
        \draw[thick,-latex] (c1\i) -- (c1\y);
      }
     }

     \foreach \i in {1,...,3}{
       \node[fill=gray] at (-2*0.5 + 0.25 + 0.5 * \i + 1.5, -2.5) (c2\i) {};
     }
    
     \foreach \i in {1,...,3}{
      \pgfmathtruncatemacro\y{Mod(\i,3)+1}
      \ifthenelse{\i=3}{
        \draw[thick,-latex] (c2\i) to [bend left=90,looseness=1.3] (c2\y);
      }{
        \draw[thick,-latex] (c2\i) -- (c2\y);
      }
     }

     \foreach \i in {1,...,6}{
      \pgfmathtruncatemacro\y{Mod(\i + 1,2) + 1}
      \draw[] (Y\i) -- (c1\y);
     }

     \foreach \i in {1,...,6}{
      \pgfmathtruncatemacro\y{(\i-1) / 2}
      \pgfmathtruncatemacro\yy{\y + 1}
      \draw[] (Y\i) -- (c2\yy);
     }

     \foreach \i in {1,...,6}{
      \pgfmathtruncatemacro\y{Mod(\i-1,3) + 1}
      \draw[] (B\i) -- (c2\y);
     }

     \foreach \i in {1,...,6}{
      \pgfmathtruncatemacro\y{(\i-1) / 3}
      \pgfmathtruncatemacro\yy{\y + 1}
      \draw[] (B\i) -- (c1\yy);
     }

     \foreach \i in {1,...,6}{
      \pgfmathtruncatemacro\y{Mod(\i-1,3) + 1}
      \draw[] (O\i) -- (c2\y);
     }

     \foreach \i in {1,...,6}{
      \pgfmathtruncatemacro\y{(\i-1) / 3}
      \pgfmathtruncatemacro\yy{\y + 1}
      \draw[] (O\i) -- (c1\yy);
     }
  \end{tikzpicture}
  \caption{Symmetry cycles couple leaf orbits across multiple branches of the selector tree. The illustration omits fake edges. In the construction, cycles do not contain directed edges, but specially colored nodes that indicate direction.}
  \label{fig:symmetry_coupling}
\end{figure}
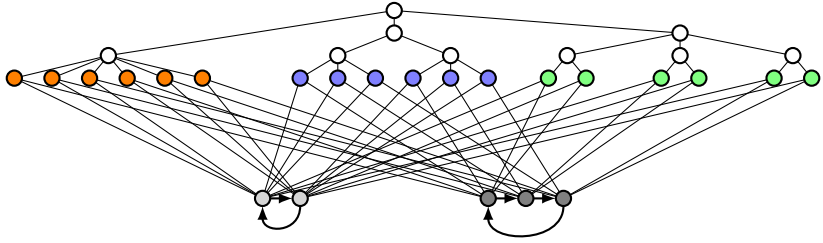

Discretization of orbits is challenging since we need to make sure that the symmetries are not destroyed by the addition of new gadgets.
Once leaf orbits have been discretized, discretization propagates through the selector tree as before and the whole construction becomes discrete.

Overall we need to construct the graph~$H_T$ mentioned in Theorem~\ref{thm:restrict_symmetries_of_tree}.  
\thmtwo*

To construct $H_T$, we introduce \emph{symmetry cycles} and \emph{symmetry couplings}. The basic idea is shown in Figure~\ref{fig:symmetry_coupling}, a detailed explanation follows below. This in turn defines a new construction $\tilde{G}(T,\pi)$ by adding a concealed version of $H_T$ to the selector tree.

\xparagraph{Discretization up to Orbits.}
Revealing the true and fake edges in $G(T,\pi)$ is not enough to discretize orbits, since this just reveals the orbit partition. By definition, nodes in the same orbit must be connected to the rest of the construction in a symmetric way and thus, splitting an orbit has to be induced by individualizations inside the orbit or through connections to other orbits that have already been split.

We thus face two independent problems related to leaf orbits. First, when the IR-process on $G(T,\pi)$ reaches leaf $l$, the orbit of $l$ may not be discrete in the current construction. Second, other leaf orbits have not been split at all. We solve these problems in an isolated setting first, by providing a constructive proof of Theorem~\ref{thm:restrict_symmetries_of_tree}. We then add the graph $H_T$ of the construction on top of $G(T,\pi)$ to obtain our final construction $\tilde{G}(T,\pi)$.

For now we are in the setting of Theorem \ref{thm:restrict_symmetries_of_tree}. We first describe how to construct $H_T$ from $T$.

\xparagraph{Symmetry Cycles.}
Consider a leaf orbit $\Omega$ in $T$.
Let~$p_1p_2\dots{}p_m = |\Omega|$~denote a prime factorization. 
We construct directed cycles of length $p_i$ for $i \in \{1, \dots{}, m\}$, such that we have one cycle for each prime $p_i$.
Cycles of the same length are ordered, which is expressed by giving them distinct colors. 

To model a directed edge, we employ two colored vertices. We add two special vertex colors $d_1, d_2$ for this purpose.
A symmetry cycle of size $n$ consists of $n$ base nodes and $2n$ edge nodes, of which $n$ are colored with $d_1$ while the other $n$ are colored with $d_2$. We define an arbitrary order on the base nodes $b_1,\dots{},b_n$, $d_1$-colored edge nodes $d_{1,1},\dots{},d_{1,n}$ and $d_2$-colored edge nodes $d_{2,1}, \dots{}, d_{2,n}$.
The cycle is then connected up by attaching $b_i$ to $d_{1,i}$, $d_{1,i}$ to $d_{2,i}$ and $d_{2,i}$ to $b_{i + 1}$ for all $i \in \{1, \dots{}, n\}$ (we set~$b_{n+1}=b_1$).

\xparagraph{Symmetry Coupling.}
The next step of the construction is to match leaf orbits with symmetry cycles (see Figure \ref{fig:symmetry_coupling}).
This naturally restricts the possible symmetries of leaf orbits but we can choose the connections in a consistent way that does not break up any orbits. 

The pairwise matching of leaf orbits is realized by coupling each orbit with the set of symmetry cycles. 
Thus, it is enough to describe a coupling between one leaf orbit $\Omega$ and the set of symmetry cycles.
To this end, we first introduce a new tree $T_\Omega$. 

Consider the common ancestor $a$ of $\Omega$ in $T$ that has least distance to $\Omega$. The root of $T_\Omega$ is $a$ and $T_\Omega$ contains exactly those $a$-to-leaf branches of $T$ that end in $\Omega$. Then $T_\Omega$ describes the group structure of the symmetries of $\Omega$ that correspond to automorphisms of $T$. Note that root-to-leaf branches of $T_\Omega$ can be permuted transitively.
In particular, the degree of $T_\Omega$ is uniform for each level. Then sibling classes on the same level have the same size and this size always divides $|\Omega|$. Note that the sibling class size may actually be $1$ for some levels. 

We modify $T_\Omega$ into another tree $T'_\Omega$ whose sibling class sizes are prime numbers. The first modification is to iteratively contract levels of $T_\Omega$ if the branching factor between them is $1$. This removes sibling classes of size $1$. Next consider the $i$-th level of $T_\Omega$ and assume the sibling class size on level $i$ is a compound number, say $s=rp$ for a prime $p$ and $r>1$. We add a new level between levels $i$ and~$i-1$ by partitioning each sibling class on level $i$ arbitrarily into $p$ classes of size $r$. We repeat the process exhaustively to obtain $T'_\Omega$.

Let $|\Omega|=p_1\cdots p_m$ be a prime factorization, then the multiset of branching factors in $T'_\Omega$ is given by $\{\!\!\{p_1,\dots,p_m\}\!\!\}$. Furthermore, each permutation of leaves corresponding to an automorphism of $T'_\Omega$ also defines an automorphism of $T$, since both types of modifications we described only restrict the possible symmetries but they do not break up orbits: contracting levels with branching factor $1$ does not interfere with automorphisms at all and when partitioning sibling classes into equally sized blocks, the action on each sibling class remains transitive. Therefore, the leaves of $T'_\Omega$ still form one orbit.

We use $T'_\Omega$ to define a coupling between leaves in $T$ and symmetry cycles.

For each sibling class $C$ on level $i$ of $T'_\Omega$, we connect the descendants of $C$ to a symmetry cycle (whose length is the sibling class size of level $i$), such that leaves are connected to the same vertex of the cycle if and only if they descend from the same node in $C$. In particular, sibling classes of leaves are connected to symmetry cycles via a perfect matching. We always use one fixed symmetry cycle for each level.
Recall that symmetry cycles of the same length are ordered. For all orbits, we always use the first cycle of length $p$ for the highest level with sibling class size $p$ and so on. This ensures that we do not introduce dependencies on rotations of different symmetry cycles (different orbits in $T$ might have ancestors in a common orbit).

\begin{proof}[Proof of Thm. \ref{thm:restrict_symmetries_of_tree}]
	We construct $H_T$ from $T$ by attaching $T'_\Omega$ to each leaf orbit $\Omega$ in $T$, such that
	we identify leaves of $T'_\Omega$ with nodes in $\Omega$. We choose a color that is not contained in $T$ to color inner vertices of $T'_\Omega$. Then we add symmetry cycles to $H_T$ (as a disjoint union) and connect the symmetry cycles with each $T'_\Omega$ as described in the construction above. Again, we use new colors for each symmetry cycle. Thereby we make sure that $\Aut(H_T)$ fixes the copy of $T$ as well as each symmetry cycle
	and each $T'_\Omega$ setwise. 
	
	By construction, $\Aut(T'_\Omega)$ acts on $\Omega$ as a transitive subgroup of $\Aut(T)|_{\Omega}$ (automorphisms restricted to $\Omega$).
	Consider the graph $H'(\Omega)$ induced by $H_T$ on $T'_\Omega$ and the set of symmetry cycles.	
	Let level $i$ of $T'_\Omega$ be connected to a symmetry cycle $C_{p_i}$.
	Observe that a rotation of $C_{p_i}$ induces a simultaneous cyclic permutation in all sibling classes on 		level $i$ and that in $H'(\Omega)$ different symmetry cycles can be rotated independently from each other.
	 Moreover, all automorphisms of $H'$ are induced by rotations of symmetry cycles
	 and since all sibling classes of $T'_\Omega$ can be permuted transitively, $\Aut(H'(\Omega))\leq \Aut(T)_{|\Omega}$ acts regularly on $\Omega$. 
	 
	Since we choose the order of symmetry cycles of the same length consistently for all orbits, we do not 			introduce dependencies between symmetry cycles, even in the full construction $H_T$. 
	This finally implies that the action of $\Aut(H_T)$ on $\Omega$ is permutation isomorphic to the action of $\Aut(H'(\Omega))$ on $\Omega$, in particular the action
	on the full set of leaves is semiregular.
\end{proof}

\xparagraph{Discretization of Orbits.}
To build $\tilde{G}(T,\pi)$, we now add the construction from Theorem $\ref{thm:restrict_symmetries_of_tree}$ to the selector tree in $G(T,\pi)$.
Observe that each leaf of $T'_\Omega$ is connected to exactly one vertex in each symmetry cycle.
That means that individualization of a leaf in $\tilde{G}(T,\pi)$ individualizes a node in each symmetry cycle and in turn, all symmetry cycles become discrete.
Moreover, since leaves that are not siblings have predecessors that are siblings in some higher level, for each pair of leaves there is one symmetry cycle such that the leaves are connected to different nodes of the cycle.
As a consequence, individualizing a leaf in $\tilde{G}(T,\pi)$ discretizes all symmetry cycles which then distinguishes all leaves from each other.
\begin{fact}\label{cor:discrete_sym_leaves}
	Individualization of a root-to-leaf path in $\tilde{G}(T,\pi)$ discretizes the set of leaves.
\end{fact}

\xparagraph{Concealing Symmetry Couplings.}
We need to hide $H_T$ from color refinement until a leaf is individualized, or otherwise leaves would be distinguishable from internal nodes in the selector tree.
As before, we do so by employing concealed edges. In the construction of $H_T$, we replace all edges with true edge gadgets.
Then, to conceal the edges, all pairs $(n,v)$, where $v$ is contained in a symmetry cycle and $n$ is a node in the selector tree, which are not yet connected by a true edge gadget, are connected with a fake edge. 
The type of these edge gadgets is revealed upon activating a (unidirectional) broadcast gadget. For this we connect the inner nodes of the concealed edge gadgets to the output of all broadcast gadgets. 

\subsection{Correctness}\label{subsection:correctness}
In this section, we prove the correctness of our graph constructions $G(T,\pi)$ and $\tilde{G}(T,\pi)$.
We start by proving some more specific properties, which ultimately culminate in our main theorem.

Throughout the section, if $v$ is a node of the selector tree in $G(T,\pi)$, then $v_T$ denotes its corresponding node in $T$.

\begin{lemma} \label{lem:colorstable}
	The selector tree in $\tilde{G}(T,\pi)$ is stable under color refinement.
\end{lemma}
\begin{proof}
	Initially, vertices in the selector tree are colored with their level.
	Recall that by our concealing paradigm, all connections of the selector tree are hidden from color refinement
	and nodes on the same level are connected to the same combined number of true or fake edges.
	This immediately implies the claim.
\end{proof}

\begin{lemma} \label{lem:asymdiscrete}
	Let $(T,\pi)$ be asymmetric and let $l$ be a leaf in the selector tree. If the concealed edges connecting $l$ to its corresponding $\AND$-gadgets have been revealed, $G(T,\pi)$ becomes discrete after individualizing $l$ and applying color refinement.
\end{lemma}
\begin{proof}
  Since we assume concealed edges to the respective $\AND$-gadgets have been revealed, individualizing $l$, by construction, splits the vertices of its corresponding individualization output.
  Since $l$ is a leaf in the selector tree, the connected broadcast gadget is a unidirectional gadget. The gadget is connected to all inputs of receiver gadgets in the graph.
  Hence, the split is propagated and all individualization outputs in the graph are split.
  
  Now, the individualization outputs in turn reveal all edges in the selector tree, as well as concealed edges to color nodes.
  Since $T$ is asymmetric, the connections to the color nodes in turn discretize nodes in the selector tree that correspond to leaves of $T$.

  Since we also reveal all edges of the selector tree itself, all nodes in the selector tree subsequently become discrete.
  This fully discretizes the attached $\AND_i$ gadgets as well as their connected individualization outputs.
  Note that at this point, for any broadcast gadget, even if they are a dead end gadget, all inputs and outputs are discrete, meaning the gadgets themselves become discrete as well.

  Since all nodes belonging to sets connected by concealed edges are now discrete, and all edges have been revealed, the concealed edge gadgets now become fully discrete as well.
  
  This in turn covers all of the constructions in $G(T,\pi)$.
\end{proof}

\begin{lemma}\label{lem:cell_choices}
	Consecutive choices of the cell selector on $G(T,\pi)$ correspond to sibling classes along paths of $T$.
\end{lemma}
\begin{proof}
  Initially, $v$ is colored with the level of $v_T$. In particular, nodes corresponding to the first level of $T$ form a color class in $G(T,\pi)$ that is stable under color refinement (see Lemma~\ref{lem:colorstable}). 
  Hence, it is by definition the first class the cell selector chooses.
  
  In case $v_T$ is a leaf of the tree, the graph becomes discrete. This implies there is no subsequently selected cell. 
  Hence, we can assume $v_T$ is an inner node of the tree.

	By definition, whenever a node $v$ is individualized, the next cell chosen by the cell selector corresponds 
  to children of $v_T$. Recall that $v$ is connected to other nodes of the selector tree via true edge 			gadgets if and only if they correspond to children of $v_T$ and $v$ is connected to all other nodes of the selector tree via fake edge gadgets. By construction, individualizing $v$ activates the individualization output of $v$.

  Since $v$ is an inner node, this split does not propagate into other gadgets:
  the receiver gadget is a unidirectional gadget in the wrong direction. This gadget therefore does not propagate the split.
  Furthermore, the broadcast gadget is a dead end gadget.
  Note that the $\AND$, receiver and broadcast gadgets attached to $v$ can be distinguished from the other gadgets of their respective type.
  However, none of these splits propagates further since all the other gadgets are connected uniformly to $v$ and the gadgets of $v$.   

  The individualization output does however reveal the edges in the selector tree that connect $v$ to its children: after individualizing $v$, fake edge gadgets attached to $v$ are distinguished from true edge gadgets attached to $v$ and so the next cell can be chosen among children of $v$. 
  
  It remains to argue that at this point, children of $v$ are indistinguishable in $G(T,\pi)$. 
  We may inductively assume that edge types between higher levels have not been revealed yet. 
  Thus, since edge types are initially indistinguishable by color refinement, the only relevant connections children of $v$ have are connections to the layer of $v$ and to inputs of $\AND_i$ gadgets. 
	Both of these connections are uniform by construction.
\end{proof}

\begin{lemma}\label{lem:diff_colors_diff_quot_graphs}
	Consider two nodes $v,w$ in the selector tree of $G(T,\pi)$. If
	$\pi(v_T)\neq\pi(w_T)$ then individualization of $v$ and $w$, respectively, produces different quotient graphs.
\end{lemma}
\begin{proof}
	Individualizing $v$ or $w$ also activates their corresponding individualization output, which in turn reveals the concealed edges which connect $v$ or $w$ to the color nodes.
	In particular, the corresponding quotient graphs already differ with respect to these connections, since $\pi(v) \neq \pi(w)$.
\end{proof}

\begin{lemma}\label{lem:orbits_of_selector_tree}
	Consider nodes $v,w$ in the selector tree of $G(T,\pi)$ such that $v_T$ and $w_T$ are leaves of $T$.
	If $\pi(v_T)=\pi(w_t)$ then $v$ and $w$ can be mapped to each other via automorphisms of  $G(T,\pi)$.
	The same holds for  $\tilde{G}(T,\pi)$.
\end{lemma}
\begin{proof}
	Recall that Lemma~\ref{lem:leaf_color_complete_inv} implies that the equally colored leaves $v_T$ and $w_T$ lie in the same orbit of $T$. The selector tree without connections to individualization outputs or symmetry coupling is just
	a concealed copy of $T$, where edges and non-edges were replaced by true edge gadgets and fake edge gadgets, respectively and colors were replaced by true/fake connections to color nodes. Thus automorphisms of $(T,\pi)$ are in one-to-one correspondence with automorphisms of the subgraph induced on the isolated selector tree together with color nodes. 
	By construction, two nodes in a common cell are connected uniformly to individualization outputs
	belonging to their cell or their common parent cell. Thus, all automorphisms of the selector tree induce 
	automorphisms of $G(T,\pi)$ by permuting individualization outputs (and the corresponding gadgets) accordingly. 
	Finally, from Theorem~\ref{thm:restrict_symmetries_of_tree} we obtain that the leaf orbits of $G(T,\pi)$ are the same as the leaf orbits of $\tilde{G}(T,\pi)$.
\end{proof}

\begin{lemma}\label{lem:same_colors_same_quot_graphs}
	Consider two nodes $v,w$ in the selector tree of $G(T,\pi)$. If
	$\pi(v_T)=\pi(w_T)$ then individualizing nodes along paths to $v$ and $w$, respectively, produces the same sequence of quotient graphs.
\end{lemma}
\begin{proof}
	First recall that due to Lemma~\ref{lem:quot_become_finer}, color classes in $T$ are contained within single layers.
  This implies that $v$ and $w$ belong to the same level $l$ of the selector tree and in particular they are connected to the inputs of $\AND_i$ gadgets uniformly.
  We make a case distinction on whether $v_T$ and $w_T$ are leaves or not. 

  Assume $v_T$ and $w_T$ are inner nodes of $T$.
  Individualizing $v$ or $w$ reveals the concealed edges connecting them to color nodes. 
  However, by assumption, they are connected to the same color node.

  Furthermore, individualizing $v$ or $w$ reveals the concealed edge gadgets connecting level $l$ to $l + 1$ in the selector tree. By Lemma~\ref{lem:quot_become_finer}, $v$ and $w$ have the same number of children on level $l + 1$. 
  Furthermore, the concealed edges connecting the children to other parts of the graph are not revealed. In particular their color has not been revealed. Hence, they are still indistinguishable.
  
  Also, due to Lemma~\ref{lem:quot_become_finer}, nodes of the same color in $T$ have predecessor nodes that are of the same color level-wise and have the same number of children. 
  In case that $v_T$ is not a leaf, the latter implies that $v$ and $w$ are uniformly connected in all steps of the construction. 
  Since edges have only been revealed up to the level of $v$ and $w$, this shows equality of quotient graphs.
  
  If $v_T$ and $w_T$ are leaves, we can apply Lemma~\ref{lem:orbits_of_selector_tree}. Note that actually the complete root-to-leaf paths for $v_T$ and $w_T$ are in the same orbit and thus individualizations along both paths produce isomorphic quotient graphs by the isomorphism invariance of color refinement.
\end{proof}

\begin{lemma}\label{lem:sym_discrete}
  Let $v$ correspond to a node in the selector tree that belongs to a leaf of $T$.
  If the concealed edges connecting $v$ to its corresponding $\AND$-gadgets have been revealed, $\tilde{G}(T,\pi)$ becomes discrete after individualizing a root-to-$v$ path and applying color refinement.
\end{lemma}
\begin{proof}
	Consider a node $l$ in the selector tree for which $l_T$ is a leaf of $T$. Recall that $\tilde{G}(T,\pi)$
	is just $G(T,\pi)$ extended by symmetry cycles and symmetry coupling. In particular, as in the asymmetric case, individualizing $l$ will reveal the colors of nodes in the selector tree (see the proof of Lemma~\ref{lem:asymdiscrete}). In particular, since leaf colors correspond to orbits, color refinement partitions the leaves into their orbits. Furthermore, all edge types are revealed at this point and thus, form a combinatorial perspective, we may treat true edge gadgets as edges and fake edge gadgets as non-edges.
	
	By Fact~\ref{cor:discrete_sym_leaves}, individualization of a path to $v$ induces the complete discretization of the set of leaf nodes in the selector tree and, as in the asymmetric case, this discretizes the whole construction.
\end{proof}

We are now ready to prove our main theorem:
\begin{theorem}\label{thm:sufficient:restated} Let $(T, \pi)$ be a colored tree that fulfills the necessary conditions. Then, $\Gamma_{S(T)}(\tilde{G}(T))$ is equal to $(T, \pi)$ (up to renaming colors).  
\end{theorem}
\begin{proof}
	By Lemma \ref{lem:colorstable}, the selector tree in $\tilde{G}(T,\pi)$	is initially stable under color refinement and in particular, its levels form stable color classes. 
	The cell selector chooses the first level of the selector tree as the first cell to individualize. 
	By Lemma \ref{lem:cell_choices}, the subsequent choices are always given by the full set of children of the node last individualized. Together with Lemma \ref{lem:sym_discrete}, this implies that the tree structure of the
	IR-tree $\Gamma:=\Gamma_{S(T)}(\tilde{G}(T))$ is exactly the same as the structure of $T$ and we obtain a one-to-one correspondence between $\Gamma$ and $T$.

	Finally, the Lemmas \ref{lem:diff_colors_diff_quot_graphs} and \ref{lem:same_colors_same_quot_graphs}
	together show that nodes in $\Gamma$ obtain the same color (i.e., the sequences of individualizations they describe give the same quotient graphs) if and only if the corresponding nodes in $T$ have the same colors, so up to renaming colors $\Gamma$ and $T$ are the same tree.
\end{proof}

\section{Conclusion and Future Work} \label{sec:conclusion}
We have shown that every tree that meets some simple necessary conditions is an IR-tree. Regarding invariant pruning we should highlight that of course every pruned tree is a subtree of an unpruned tree, so our techniques extend to IR-algorithms with pruning.

Regarding refinement, we use the standard color refinement used by all IR-algorithms. 
However regarding cell selectors there is no clear standard. 
In this paper, we did not optimize the construction for any specific cell selector, but rather used the cell selector as part of the construction.

Let us now assume we are given a fixed cell selector.
For a particular cell selector, there are two possibilities: either, fewer trees turn out to be IR-trees or the same necessary conditions apply.
For the latter, we suspect that for many natural examples the construction of this paper can be adapted.
Consider for example the cell selector that always chooses a smallest non-trivial cell. In this case, by adding more concealed structure enforcing specific cell sizes it can be shown that the same necessary conditions are indeed sufficient again. 

In contrast to this, consider the cell selector that always chooses a largest non-trivial cell. Here, the degree of the vertices on root-to-leaf walks in a corresponding IR-tree must monotonically decrease. Hence, fewer trees turn out to be IR-trees and the necessary conditions are not sufficient. If interested in specific cell selectors one might therefore want to refine the necessary conditions.

Another interesting direction of research might be to investigate bounds for the order graphs realizing a given tree since this is related to the running time of IR-tools.

\bibliography{main}
\bibliographystyle{plain}
\end{document}